\documentclass[11pt]{article}

  \usepackage{fullpage}
  \usepackage{amsthm, latexsym,amssymb,amsfonts,amsmath,mathrsfs,float}
  \usepackage[font=small,labelfont=bf, width=.618\textwidth]{caption} 
  \usepackage[utf8]{inputenc} 
  \usepackage[T1]{fontenc}
  \usepackage[dvipsnames]{xcolor}
  \usepackage{graphicx}
  \usepackage[linktocpage=true, linkbordercolor=red, ocgcolorlinks]{hyperref}
  \hypersetup{colorlinks   = true,citecolor    = Cerulean,urlcolor = WildStrawberry }
  \usepackage{enumitem}
  \makeatletter
  \def\namedlabel#1#2{\begingroup
	#2%
	\def\@currentlabel{#2}%
    	\phantomsection\label{#1}\endgroup
  }

  \usepackage{tikz, color}
	\usetikzlibrary{shapes}
	\usetikzlibrary{shapes.arrows}

  \newcommand{\ignore}[1]{}	

  \newcommand{\peq}{{\sc EQ}}
  
  \newcommand{\lmin}{\lambda_{\text{\sc min}}}

  \newcommand{\R}{\mathbb{R}} 
  \newcommand{\C}{\mathbb{C}} 
  \newcommand{\N}{\mathbb{N}} 
  \newcommand{\Z}{\mathbb{Z}} 
  \newcommand{\bset}[1]{\{0,1\}^{#1}} 
  
  \newcommand{\ket}[1]{|#1\rangle}
  \newcommand{\bra}[1]{\langle#1|}
  \newcommand{\braket}[2]{\langle #1|#2\rangle}
  \newcommand{\ketbra}[2]{|#1\rangle\!\langle#2|}

  \newcommand{\ceil}[1]{\lceil{#1}\rceil}
  
  \DeclareMathOperator{\Tr}{\mathsf{Tr}}

  \newcommand{\st}{:\,} 
  
  \DeclareMathOperator{\mon}{\mathbf{mon}}
  \newcommand{\eps}{\varepsilon}  

  \newcommand{\ie}{{i.e.}}

  \newcommand{\beq}{\begin{equation}}
  \newcommand{\eeq}{\end{equation}}
  \newcommand{\beqn}{\begin{equation*}}
  \newcommand{\eeqn}{\end{equation*}}
  \newcommand{\beqr}{\begin{eqnarray}}
  \newcommand{\eeqr}{\end{eqnarray}}
  \newcommand{\beqrn}{\begin{eqnarray*}}
  \newcommand{\eeqrn}{\end{eqnarray*}}
  \newcommand{\bmline}{\begin{multline}}
  \newcommand{\emline}{\end{multline}}
  \newcommand{\bmlinen}{\begin{multline*}}
  \newcommand{\emlinen}{\end{multline*}}

  \newtheorem{theorem}{Theorem}[section]

  \newtheorem{proposition}[theorem]{Proposition}
  \newtheorem{lemma}[theorem]{Lemma}
  \newtheorem{conjecture}[theorem]{Conjecture}

  \newtheorem{corollary}[theorem]{Corollary}
  \newtheorem{definition}[theorem]{Definition}

  \makeatletter
\renewcommand*{\@fnsymbol}[1]{\ensuremath{\ifcase#1\or *\or \mathparagraph\or \ddagger\or
    \mathsection\or \dagger\or \|\or **\or \dagger\dagger
    \or \ddagger\ddagger \else\@ctrerr\fi}}
\makeatother
  
\begin{document}

\title{
\bf
Round elimination in exact communication complexity
}

\author{
Jop Bri\"{e}t\thanks{
Centrum Wiskunde \& Informatica (CWI), Science Park 123, 1098 XG Amsterdam, The Netherlands.
Funded by a Rubicon grant from the Netherlands Organization for Scientific Research (NWO).
E-mail: \texttt{j.briet@cwi.nl}.
}
\and
Harry Buhrman\thanks{
QuSoft, University of Amsterdam, and CWI, Science Park 123, 1098 XG Amsterdam, The Netherlands.
Supported by NWO Gravitation-grants NETWORKS and QSC as well as EU grant QuantAlgo and CIFAR.
E-mail: \texttt{buhrman@cwi.nl}.
}
\and
Debbie Leung\thanks{
University of Waterloo, Waterloo, Canada.
Funded in part by NSERC, CRC, and CIFAR.  
E-mail: \texttt{wcleung@math.uwaterloo.ca}.
}
\and
Teresa Piovesan\thanks{
Centrum Wiskunde \& Informatica (CWI), Science Park 123, 1098 XG Amsterdam, The Netherlands.
Funded in part by the EU project SIQS.
E-mail: \texttt{tere.piovesan@gmail.com}.
}
\and
Florian Speelman\thanks{
Qusoft and Centrum Wiskunde \& Informatica (CWI), Science Park 123, 1098 XG Amsterdam, The Netherlands.
Funded in part by the EU project SIQS.
E-mail: \texttt{speelman@cwi.nl}.
}
}

\maketitle

\date{}

\begin{abstract}
We study two basic graph parameters, the chromatic number and the orthogonal rank, in the context of classical and quantum exact communication complexity.
In particular, we consider two types of communication problems that we call \emph{promise equality} and \emph{list} problems.
For both of these, it was already known that
the one-round classical and one-round quantum complexities are characterized by the chromatic number and orthogonal rank of a certain graph, respectively.

In a promise equality problem, Alice and Bob must decide if their inputs are equal or not.
We prove that classical protocols for such problems can always be reduced to one-round protocols with no extra communication.
In contrast, we give an explicit instance of a promise equality problem that exhibits an exponential gap between the one- and two-round exact quantum communication complexities.
Whereas the chromatic number thus fully captures the complexity of promise equality problems, the hierarchy of ``quantum chromatic numbers'' (starting with the orthogonal rank) giving the quantum communication complexity for every fixed number of communication rounds turns out to enjoy a much richer structure.

In a list problem, Bob gets a subset of some finite universe, Alice gets an element from Bob's subset, and their goal is for Bob to learn which element Alice was given.
The best general lower bound (due to Orlitsky) and upper bound (due to Naor, Orlitsky, and Shor) on the classical communication complexity of such problems differ only by a constant factor.
We exhibit an example showing that, somewhat surprisingly, the four-round protocol used in the bound of Naor et al.\ can in fact be optimal.
Finally, we pose a conjecture on the orthogonality rank of a certain graph whose truth would imply an intriguing impossibility of \emph{round elimination} in quantum protocols for list problems, something that works trivially in the classical case.
\end{abstract}

\newpage

\section{Introduction}

The chromatic number~$\chi(G)$ of a graph~$G$ is the minimum number of colors needed to color the vertices in such a way that adjacent vertices get different colors.
This important graph parameter appears frequently in computer science
and mathematics; it is well-known to be NP-hard to compute and has recently found a number of meaningful generalizations in the context of non-local games and entanglement-assisted zero-error information theory.
One of those generalizations is the \emph{orthogonal rank} of a graph, denoted~$\xi(G)$ and defined as follows.
An orthonormal representation of a graph is an assignment of complex unit vectors to the vertices such that adjacent vertices receive orthogonal vectors.
The orthogonal rank is the minimum dimension of such a representation.
Similar to the chromatic number, the orthogonal rank is also NP-hard to compute (see Appendix~\ref{app:hardrank}).
In this paper, we study both of these graph parameters in the context of communication complexity.\\

\noindent{\bf Classical communication complexity.}
Since its introduction by Yao~\cite{Yao:1979} communication complexity has become a standard model in computational complexity that enjoys a wide variety of connections to other areas in theoretical computer science~\cite{Kushilevitz:1997}. 
Here two parties, Alice and Bob, receive inputs $x,y$ from sets~$\mathcal X,\mathcal Y$ (resp.) and need to compute the value $f(x,y)$ of a two-variable function~$f$ known to them in advance. 
Usually each party has insufficient information to solve the problem alone, meaning that they have to exchange information about each others' inputs.
The idea that communication is expensive motivates the study of the {\em communication complexity} of~$f$, which counts the minimal number of bits that the parties must exchange on worst-case inputs.
Throughout this paper, we consider only exact (deterministic) communication protocols, meaning that no error is allowed, and we will omit the word \emph{exact} from now on.
Of particular importance to this paper is the distinction between {\em one-round} protocols, where all communication flows from Alice to Bob, and {\em multi-round} protocols, where they take turns in sending messages from one party to the other.
\medskip

\noindent{\bf Quantum communication complexity.}
{In yet another celebrated paper, Yao~\cite{Yao:1993} introduced {\em quantum communication complexity}, where to compute the value~$f(x,y)$ the parties are allowed to transmit {\em qubits} back and forth.
The study of this model has  also become a well-established discipline in theoretical computer science and quantum information theory.
The most basic question that arises when considering the classical and
quantum models is whether they are actually substantially different.
An upper bound on the possible difference between these models was proved by Kremer~\cite[Theorem~4]{Kremer:1995}.\footnote{The result stated here is actually a slight generalization of Kremer's result (which focuses on boolean functions) that can be proved in the same way; for completeness we give a proof in Appendix~\ref{sec:Kremer}. Moreover, this statement (as well as Kremer's original formulation) holds in the bounded-error model of communication complexity, not only in the exact one.}

\begin{theorem}[Kremer]\label{thm:kremer}
Any quantum protocol that uses $\ell$ qubits of communication can be turned into a $2^{O(\ell)}$-bit one-round classical protocol for the same problem.
\end{theorem}

The first large gap between exact classical and quantum communication complexity was demonstrated by Buhrman, Cleve, and Wigderson~\cite{Buhrman:1998}, who gave a 
problem admitting a one-round quantum protocol that is exponentially more efficient than any (multi-round) classical protocol.\\

The chromatic number and orthogonal rank naturally show up in two types of communication problems that we call \emph{promise equality} and \emph{list} problems, discussed next.

\subsection{Promise equality}

In a \emph{promise equality problem}, Alice and Bob are either given equal inputs or a pair of distinct inputs from a subset~$\mathcal D$ of $\binom{\mathcal X}{2}$ ($\mathcal D$ is known to them in advance). 
Their goal is to decide whether their inputs are equal or different. \\

\noindent{\bf Classical communication complexity of promise equality problems.}
It was observed by de~Wolf~\cite[Theorem~8.5.1]{deWolf:PhD} that if $G = (\mathcal X,\mathcal D)$ is the graph with vertex set~$\mathcal X$ and edge set~$\mathcal D$, then the one-round classical communication of the problem equals $\ceil{\log \chi(G)}$.
Analogously, for each positive integer~$r$ one can define a ``level-$r$'' chromatic number of the graph corresponding to the communication complexity of protocols that proceed in~$r$ rounds or less.
For general communication problems, using more rounds can decrease the total communication, as is the case for the general Pointer Jumping Problem, where for every positive integer $m$ there is an instance for which any $m$-round protocol requires exponentially more communication than the best $(m+1)$-round protocol~\cite[{Section 4.2}]{Kushilevitz:1997}. 
However, we show that this is not true for promise equality problems (Lemma~\ref{lem:CCeq} below),
meaning that for such problems
the chromatic number not only characterizes the one-round complexity, but their overall communication complexity.\\

\noindent{\bf Quantum communication complexity of promise equality problems.}
The one-round quantum communication complexity of promise equality problems is characterized by the orthogonal rank of the associated graph~$G = (\mathcal X, \mathcal D)$~\cite[Theorem~8.5.2]{deWolf:PhD}.
Indeed, it is not difficult to see that a one-round
quantum protocol of a promise equality problem is equivalent to an
orthonormal representation of $G$; the vectors correspond to the states that Alice would send to Bob and orthogonality is required for Bob's measurement to tell whether they got equal inputs or not.
Viewing the orthogonal rank as the ``one-round quantum chromatic number'' naturally leads one to define a hierarchy of such numbers where the level-$r$ quantum chromatic number corresponds to the communication complexity of $r$-round quantum protocols.
One might expect that, as in the classical case, this hierarchy is redundant in that the levels all carry the same number.
However, one of our main results shows that in the quantum setting, this is {\em not} the case.

\begin{theorem}\label{thm:promise-general}
There exist absolute constants~$c,C\in (0,\infty)$, an infinite sequence of natural numbers $n \in \N$ and a family of promise equality problems $( \{0,1\}^{n},\mathcal D_n)_{n\in\N}$ such that:
\begin{itemize}
\item The one-round quantum communication complexity of~$(\{0,1\}^{n},\mathcal D_n)$ is at least~$cn$.
\item There is a two-round quantum protocol for~$(\{0,1\}^{n},\mathcal D_n)$ using at most~$C\log n$ qubits.
\end{itemize} 
\end{theorem}

During our analysis of the particular promise problem used for Theorem~\ref{thm:promise-general} we answer an open question of Gruska, Qiu, and Zheng~\cite{Gruska:2014}. 
To explain this, we briefly elaborate on what goes into our result.
The problem we consider is simple: Let~$n$ be a positive integer multiple of~$8$. Alice and Bob are given $n$-bit strings~$x$ and~$y$, respectively, that are either equal or differ in exactly~$n/4$ coordinates and they must distinguish between the two cases.
We denote this problem by \peq-$\binom{n}{n/4}$. 
Similar promise equality problems were studied before.
In the above-mentioned result of Buhrman, Cleve, and Wigderson~\cite{Buhrman:1998}, which showed the first exponential gap between classical and quantum communication, they used the problem \peq-$\binom{n}{n/2}$, where Alice and Bob get $n$-bit strings that are either equal or differ in exactly half of the entries (for $n$ a multiple of~4). 
They used a distributed version of the Deutsch--Jozsa algorithm to give a one-round $O(\log n)$-qubit quantum protocol, while a celebrated combinatorial result of Frankl and R\"{o}dl~\cite{Frankl:1987} implies that the classical communication complexity is at least~$\Omega(n)$.
Similar results were shown (based on similar techniques) in the above-mentioned paper~\cite{Gruska:2014} for the analogous problem \peq-$\binom{n}{ \alpha n}$ when $\alpha >1/2$, and the authors pose as an open problem to determine the quantum communication complexity of \peq-$\binom{n}{ \alpha n}$ when~$\alpha < 1/2$.
An easy observation is that the problem \peq-$\binom{n}{d}$ where $n$ and $d$ have different parities is trivial: Alice has just to send the parity bit of her string to Bob. For this reason, in this paper and in the above-mentioned works both $n$ and $d$ are assumed to be even numbers. 

To prove Theorem~\ref{thm:promise-general}, 
we show that the one-round quantum communication complexity of \peq-$\binom{n}{n/4}$ is at least~$\Omega(n)$ and we give a two-round protocol for it that uses at most~$O(\log n)$ qubits.
For the proof of the first bound we use the famous Lov\'asz theta number, which lower bounds the orthogonal rank and therefore the one-round quantum communication complexity.
We prove a lower bound on the theta number using the theory of association schemes and known properties of the roots of the Krawtchouk polynomials.
Our two-round protocol is based on a distributed version of Grover's  algorithm.
With a little extra technical work our results can be extended to any of the problems \peq-$\binom{n}{ \alpha n}$ with constant $\alpha < 1/2$.
In light of Kremer's Theorem and the obvious fact that the one-round classical communication complexity is at least  its quantum counterpart, we thus settle the question of~\cite{Gruska:2014}.

\subsection{The list problem}
In the \emph{list problem}, inputs are picked from a subset~$\mathcal D\subseteq \mathcal X\times\mathcal Y$ and the goal is for Bob to learn Alice's input.
The reason for the name ``list problem'' is that
Bob's input~$y$ may just as well be given to him as the list (subset) of all of Alice's possible inputs~$x$ satisfying~$(x,y)\in\mathcal D$.
A list problem can thus equivalently be given by a family~$\mathcal L\subseteq  2^{\mathcal X}$  of lists, where Bob gets a list~$L\in \mathcal L$, Alice gets an element~$x\in L$, and Bob must learn~$x$.
We refer to this communication problem as~{\sc $\mathcal L$-list}.\\

\noindent{\bf Classical communication complexity of list problems.}
Witsenhausen~\cite{Witsenhausen:1976} observed that the one-round classical communication complexity of the list problem is characterized by the chromatic number of the graph with vertex set~$\mathcal X$ and whose edge set consists of the pairs of distinct elements appearing together in some list~$L\in \mathcal L$.
Denoting this graph by~$G_\mathcal L$, the one-round communication complexity equals~$\ceil{\log\chi(G_\mathcal L)}$.
The multi-round communication complexity of the list problem has also been studied.
Orlitsky~\cite[Corollary~3 and Lemma~3]{Orlitsky:1990} proved the following lower bound in terms of the chromatic number of~$G_\mathcal L$, and the cardinality of the largest list, denoted $$\omega(\mathcal L)= \max\{|L|\st L\in\mathcal L\}$$ (not to be confused with the cardinality of the largest clique~$\omega(G_\mathcal L)$, which can be larger).

\begin{theorem}[Orlitsky]\label{thm:lb-classical}
For every family~$\mathcal L\subseteq 2^{\mathcal X}$, the classical communication complexity of {\sc $\mathcal L$-list} is at least $\max\{ \log \log \chi(G_\mathcal L), \log \omega(\mathcal L) \}$.
\end{theorem}

The basic idea behind the above result is that any multi-round protocol can be simulated by a one-round protocol with at most an exponential difference in communication, and that Alice must send sufficient information for Bob to be able to distinguish among $\omega(\mathcal L)$ elements.
In the same work,
Orlitsky~\cite[Theorem 4]{Orlitsky:1990} gave a two-round classical protocol based on perfect hashing functions that nearly achieves the above lower bound.

\begin{theorem}[Orlitsky]\label{thm:2-rounds}
For every family~$\mathcal L\subseteq 2^{\mathcal X}$, the two-round classical communication complexity of {\sc $\mathcal L$-list} is at most ${\log \log \chi(G_\mathcal L) + 3 \log \omega(\mathcal L)+ 4}$.
\end{theorem}

It thus follows from
Witsenhausen's observation and Theorem~\ref{thm:2-rounds} that list problems have exponentially more efficient two-round protocols than one-round protocols, provided that $\omega(\mathcal L) \leq poly(\log \chi(G_{\mathcal L}))$.
But Theorem~\ref{thm:lb-classical} shows that---in stark contrast with the Pointer Jumping Problem---using more than two rounds cannot decrease the total amount of communication by more than a factor of~4, since obviously
${\log \log \chi(G_\mathcal L) + 3 \log \omega(\mathcal L)}
\leq
4\max\{ \log \log \chi(G_\mathcal L), \log \omega(\mathcal L) \}
$.
Furthermore, in a follow up work, Orlitsky~\cite{Orlitsky:1991} showed that in general two-round protocols are not sufficient to reach the communication complexity.
The natural question that thus arises is:
\begin{center}
\emph{Can the lower bound of Theorem~\ref{thm:lb-classical}  be attained by using more than two rounds of communication?}
\end{center}
Towards answering this question, Naor, Orlitsky, and Shor~\cite[Corollary 1]{Naor:1993} slightly improved on Theorem~\ref{thm:2-rounds} and showed that the four-round communication complexity gets to within a factor of about~3 of the lower bound.

\begin{theorem}[Naor--Orlitsky--Shor]\label{thm:4-rounds}
For every family~$\mathcal L\subseteq 2^{\mathcal X}$, the four-round classical communication complexity of {\sc $\mathcal L$-list} is at most $\log \log \chi(G_\mathcal L) + 2 \log \omega(\mathcal L) + 3 \log \log \omega(\mathcal L)+ 7$.
\end{theorem}

As our contribution to this line of work we show that, perhaps surprisingly, for some list problems the four-round protocol of Naor, Orlitsky, and Shor is in fact asymptotically optimal, thus answering the above question in the negative.

\begin{theorem}\label{thm:list-general}
For any~$\eps >0$ there exists a set~$\mathcal X$ and a family~$\mathcal L\subseteq 2^{\mathcal X}$ such that the classical communication complexity of {\sc $\mathcal L$-list} is at least $\log\log \chi(G_\mathcal L) + (2-\eps)\log\omega(\mathcal L)$.
Moreover, there exists such an~$(\mathcal X,\mathcal L)$ pair for which~$\omega(\mathcal L) = \log\chi(G_\mathcal L)$.
\end{theorem}

In particular, our result gives a family of list problems with communication complexity at least $(3-\eps)\max\{\log\log\chi(G_\mathcal L), \log\omega(\mathcal L)\}$ for any~$\eps>0$.\\

\noindent{\bf Quantum communication complexity of list problems and quantum round elimination.}
The one-round quantum communication complexity of list problems is given by~$\ceil{\log\xi(G_{\mathcal L})}$, which follows from the same considerations as for the promise equality problems (see Lemma~\ref{lem:list-1-quantum}).
Based on a conjecture we make about the orthogonal rank of a certain family of graphs, we believe that in the context of quantum communication complexity, list problems may have the interesting property of resisting a quantum analogue of \emph{round elimination}.

In classical communication complexity, round elimination reduces the number of rounds of a given protocol by having the parties send some extra information instead.
Consider the following basic example, where we start with a two-round $(\log n+1)$-bit protocol in which Bob starts by sending Alice a single bit and Alice replies with an $\log n$-bit string.
This protocol can easily be turned into a \emph{one-round} $2\log n$-bit protocol by having Alice directly send Bob two $\log n$-bit strings, one corresponding to the case where Bob sends a~0 in the two-round protocol and another for if he sends a~1.
Then Bob can just pick the string corresponding to the bit he would have sent based on his input and solve the problem.

A quantum analogue of the above example would turn a two-round $(\log n+1)$-qubit protocol into a one-round~$2 \log n$-qubit protocol.
We conjecture that the following family of list problems is a counterexample to the existence of such an analogue.
For an even positive integer~$n$ and~$d\in[n]$, let $\mathcal L_d \subseteq 2^{\bset{n}}$ be the family of lists~$L\subseteq \bset{n}$ of maximal cardinality such that all strings in~$L$ have Hamming distance exactly~$d$.
Consider the family of lists given by~$\mathcal K = \mathcal L_{n/2}\cup\cdots\cup\mathcal L_{n}$.
Similar to the classical example above, we give a simple two-round protocol for {\sc $\mathcal K$-list}.

\begin{theorem}\label{thm:simple2rnd}
For~$\mathcal K = \mathcal L_{n/2}\cup\cdots\cup\mathcal L_{n}$, there exists a two-round protocol for {\sc $\mathcal K$-list} where Bob sends Alice a single qubit and Alice replies with a $( (\log n)+1)$-qubit message.
\end{theorem}

It is easy to see that the graph~$G_\mathcal K = (\bset{n}, E)$ associated with~$\mathcal K$ has edge set~$E$ given by all pairs of strings with Hamming distance in~$\{n/2, \dots, n\}$.
In Appendix~\ref{sec:app-orth} we show that $\Omega(n) \leq \xi(G_{\mathcal K}) \leq O(2^{0.81 n})$.
We conjecture, however, that the orthogonal rank of the graph~$G_{\mathcal K}$ satisfies a much stronger lower bound.\footnote{This conjecture has since been confirmed by Bri\"et and Zuiddam~\cite{Briet:2017}.}

\begin{conjecture}
The graph~$G_\mathcal K$ as above satisfies $\xi(G_\mathcal K) \geq n^{\omega(1)}$.
\end{conjecture}

By the relation between the one-round quantum communication complexity of list problems and the orthogonal rank of their associated graphs, it follows that the validity of the above conjecture would imply that the exact one-round quantum communication complexity of the above problem is super-logarithmic in~$n$, in marked contrast with the classical example of round elimination.

\subsection{Connections to other work}
Our work strengthens a link between communication complexity and graph theory established by de~Wolf~\cite{deWolf:PhD}.
Orthonormal representations appear in the context of zero-error information theory.
Indeed they were introduced by Lov\'asz~\cite{Lovasz:1979} to settle a famous problem of Shannon concerning the (classical) capacity of the 5-cycle and they serve as proxies for entanglement-assisted schemes~\cite{Cameron:2007, Cubitt:2010, Leung:2012, Briet:2012, Briet:2015, Cubitt:2013}.
They also appear in the context of non-local games~\cite{Cameron:2007, Godsil:2008, Scarpa:2012}.
Nevertheless the orthogonal rank is poorly understood.
To the best of our knowledge, our result is the first time a {\em lower bound} on the dimension was used.
The use of the Lov\'asz theta number in the context of communication complexity problems also appears to be new and we hope that it may find further applications there in the future.
Finally, quantum variants of the chromatic number that appeared in for example non-local games~\cite{Cameron:2007,Scarpa:2012} and zero-error information theory~\cite{Briet:2015, Cubitt:2013} can be interpreted as quantum communication complexities of promise equality problems in various different communication models,  which puts those parameters in a more unified framework.\\

\paragraph{Outline.}
In Section~\ref{sec:equality} we study the promise equality problem and in particular we prove Theorem~\ref{thm:promise-general}.
In Section~\ref{sec:list} we discuss the list problem and prove Theorem~\ref{thm:list-general} and Theorem~\ref{thm:simple2rnd}.

\section{Preliminaries}

\paragraph{Graph theory.}
Throughout the paper we consider only graphs that are simple - undirected, no multiple edges and without loops. All the logarithms are in base two, unless specified otherwise.

Given a graph $G = (V,E)$, $V$ and $E$ denote the vertex and edge set respectively (equivalently denoted by $V(G)$ and $E(G)$).
The complement of a graph $G$, denoted by $\overline{G}$, has the same vertex set as the original graph and a pair of vertices is adjacent if and only if it is non adjacent in $G$.
The \emph{adjacency matrix} of a graph $G$ is the $|V(G)|\times |V(G)|$ symmetric matrix where the $i,j$-th entry is equal to 1 if $ij \in E(G)$ and to 0 otherwise.

A clique of a graph is a set of pairwise adjacent vertices. 
The \emph{clique number} $\omega(G)$ is the maximum cardinality of a clique.
An independent set is a set of pairwise non-adjacent vertices. The \emph{independence number}  $\alpha(G)$ is the largest cardinality of an independent set. 
In other words, for any graph $G$ we have $\omega(G) = \alpha(\overline{G})$.
A $t$-coloring of a graph $G$ is an assignment of $t$ colors such that adjacent vertices receive different colors.
The \emph{chromatic number} $\chi(G)$ is the minimum number $t$ such that a $t$-coloring exists. 
Equivalently, a coloring is a partition of the vertex set into indepedent sets and therefore the inequality $\chi(G)\alpha(G) \ge |V(G)|$ holds for every graph $G$.
A (complex) $d$-dimensional orthonormal representation of a graph $G$ is a map $\phi$ from the vertex set to the $d$-dimensional complex unit sphere such that adjacent vertices are mapped to orthogonal vectors.
The \emph{orthogonal rank} $\xi(G)$ is the minimum $d \in \mathbb{N}$ for which a $d$-dimensional orthonormal representation exists. 
We stress that we consider the representations over the complex sphere and not, as more usual in the combinatorial literature, over the real one and that orthogonalities are required for adjacent vertices. A simple observation gives that $\omega(G) \leq \xi(G)$.

With $H(n,d)$ we denote the graph that has $\{0,1\}^n$ as vertex set and where two $n$-bit strings are adjacent if they differ exactly in $d$ positions.
Equivalently, $H(n,d)$ is the graph with vertex set $\{-1,1\}^n$ where two vertices are adjacent if their inner product is equal to $n - 2d$.

\paragraph{Orthogonality lemma.}
We will make repeated use of the following standard lemma (see for example~\cite{Briet:2015}), which we will refer to here as the {\em Orthogonality Lemma}.

\begin{lemma}[Orthogonality Lemma]\label{lem:orth}
Let~$\rho_1,\dots,\rho_\ell\in\C^{d\times d}$ be a collection of Hermitian positive semidefinite matrices. Then the following are equivalent:
\begin{enumerate}
	\item We have $\rho_i\rho_j = 0$ for every~$i\ne j \in [\ell]$.
	\item There exists a measurement consisting  of positive semidefinite matrices~$P^1,\dots,P^\ell$, $P^\perp\in\C^{d\times d}$ such that
$\Tr(P^i\rho_j) = \delta_{ij}\Tr(\rho_j)$ and  $\Tr(P^\perp \rho_j) = 0$ for every~$i,j\in[\ell]$.
\end{enumerate} 

In particular, a collection of pure states~$\ket{\phi_1},\dots,\ket{\phi_\ell}\in\C^d$ can be perfectly distinguished with a quantum measurement if and only if they are pairwise orthogonal.
\end{lemma}

\section{Promise Equality}\label{sec:equality}

In a promise equality problem, Alice and Bob each receive an input from a set $\mathcal X$ with the promise that their inputs
either are equal or come from a subset $\mathcal D$ of $\binom{\mathcal X} {2}$ (known to the players beforehand). The goal is to distinguish between the two cases.
To any promise equality problem, we associate the graph $G= (\mathcal X, \mathcal D)$.

\subsection{General properties of promise equality}

As we mentioned earlier, the one-round classical communication complexity of the problem equals $\ceil{\log \chi(G)}$.
We begin by proving that the chromatic number of the associated graph actually gives the overall communication complexity.

\begin{lemma}\label{lem:CCeq}
For any promise equality problem, the classical communication complexity is attained with a single round of communication.
\end{lemma}

\begin{proof}
We show how to transform a $k$-round communication protocol into a one-round protocol that uses the same amount of bits.
To summarize, the idea is that Alice mimics all the rounds of communication assuming that her input is equal to Bob's one, and sends them in one-round.
He then checks whether the message received is consistent with his input. If this is not the case, then he knows that the two strings are different, otherwise he completes the protocol.

More formally, fix a protocol $P$ that requires $k$ rounds, where $k \ge 2$. 
Suppose that Alice has input $x$ and Bob has $y$.
We assume that the first round of communication is from Alice to Bob, but the same reasoning applies in the other case.
For $i$ odd, let $a_i$ be the message that Alice would send to Bob on the $i$-th round of communication if she followed protocol $P$ and used the knowledge of the messages exchanged in the previous rounds and of her input $x$. 
Similarly, for $i$ even, let $\hat{b}_i$ be the message that Bob would send to Alice on the $i$-th round of communication if he had $y=x$ as input, followed the protocol $P$ and used the knowledge derived by the previous rounds. 
Using the protocol $P$, Alice can mimic Bob's rounds of communication under the assumption that Bob's input is equal to $x$.
Alice uses her input $x$ to produce the string $a_1 \hat{b}_2 a_3 \dots a_i \hat{b}_{i+1} \dots a_k$ and sends it to Bob in one round.
From his input $y$, Bob constructs the messages $b_i$ that he would have produced during the protocol $P$, with the knowledge of Alice's messages $a_\ell$ and his messages $b_\ell$ for all $\ell<i$.
If there exists an index $i$ such that $b_i \neq \hat{b}_i$, then $x$ must be different from $y$. Otherwise, Bob uses the transcript $a_1 \hat{b}_2 a_3 \dots a_i \hat{b}_{i+1} \dots a_k$ to finish the protocol and either outputs $x = y$ or $x \neq y$.
We have constructed a one-round communication protocol $\hat{P}$ that works as the original protocol $P$ does and whose worst-case transcript length is at most as long as the one of $P$.
Therefore if $P$ is an optimal protocol, so is~$\hat{P}$.
\end{proof}

As mentioned in the introduction, de~Wolf~\cite[Theorem~8.5.2]{deWolf:PhD} showed that one-round quantum protocols are related to orthonormal representations.
We include a proof here for completeness.

\begin{theorem}[de~Wolf] \label{thm:1roundquantum}
Consider a promise equality problem defined by the sets $\mathcal X$ and $\mathcal D$, then its one-round quantum communication complexity is equal to $\ceil{\log \xi(G)}$, where $G = (\mathcal X, \mathcal D)$. 
\end{theorem}

\begin{proof}
Let $P$ be an optimal one-round protocol for the considered promise equality problem.
Without loss of generality, Alice sends pure state $\ket{\phi_x} \in \mathbb{C}^d$ on input $x \in \mathcal X$.
For any pair $(x,y) \in \mathcal D$, $\ket{\phi_x}$ and $\ket{\phi_y}$ have to be perfectly distinguishable  and therefore, in view of Lemma~\ref{lem:orth}, they must be orthogonal.
Hence, the map $\phi : \mathcal X \to \mathbb{C}^d$ where $\phi(x) = \ket{\phi_x}$ is a $d$-dimensional orthonormal representation of $G = (\mathcal X, \mathcal D)$ and $\xi(G) \le d$.

On the other hand, let $\phi$ be a $d$-dimensional orthonormal representation of the graph $G = (\mathcal X, \mathcal D)$ and consider the one-round quantum protocol that transmits $\phi(x) \in \mathbb{C}^{d}$ on input $x \in \mathcal X$.
This uses $\log d$-qubits of communication.
From Lemma~\ref{lem:orth} we know that Bob can use his input $y$ to perform a quantum measurement that allows him to learn whether his input is equal or not to Alice's one. 
Thus, the one-round quantum communication complexity of this  problem is at most
$\ceil{\log \xi(G)}$.
\end{proof}

\subsection{Proof of Theorem~\ref{thm:promise-general}}

The rest of this section will be devoted to the proof of Theorem~\ref{thm:promise-general}, which shows that there is a family of promise equality problems where allowing two rounds of quantum communication is exponentially more efficient than a single round.
The problem that exhibits this separation is \peq-$\binom{n}{n/4}$, where Alice and Bob each receive an $n$-bit string that are either equal or differ in exactly $n/4$ positions (with $n$ a multiple of $8$).
We denote by $H(n,n/4)$ the graph associated with this problem.
We split the proof in two parts: in Section~\ref{sec:1qc} we bound the one-round quantum communication complexity and in Section~\ref{sec:2qc} we give the two-round protocol.

With a bit of extra work, we can also prove that for any of the problems \peq-$\binom{n}{\alpha n}$ where $\alpha \in (0, 1/2)$ and $n$, $\alpha n$ are even, there is a quantum multi-round protocol that is exponentially more efficient than a single round one. The lower bound on the one-round quantum communication complexity is explain in Section~\ref{sec:1qc}, while the multi-round protocol can be found in Appendix~\ref{sec:exact-grover-general}.

\subsubsection{One-round quantum communication complexity of \texorpdfstring{\sc \peq-$\binom{n}{n/4}$}{EQ-(n n/4)}}\label{sec:1qc}

The main result of this section is the following theorem, which gives the first part of Theorem~\ref{thm:promise-general}. 

\begin{theorem}\label{thm:xi-smaller12}
The one-round quantum communication complexity of \peq-$\binom{n}{n/4}$ is at least~$\Omega(n)$.
\end{theorem}

We remark that the statement of the above theorem holds for any problem \peq-$\binom{n}{ \alpha n}$ where $\alpha \in (0,1/2)$ and where $n$, $\alpha n$ are even. Indeed, in this section we prove a lower bound on the orthogonal rank of any graph $H(n, \alpha n)$ (where $\alpha \in (0,1/2)$ and $n$, $\alpha n$ are even) and then use Theorem~\ref{thm:1roundquantum} and focus on the case where $\alpha = 1/4$ to derive Theorem~\ref{thm:xi-smaller12}.

We prove the desired bound in three steps: 
first, we show that the Lov\'asz theta number is a lower bound for the orthogonal rank; 
second, we use structural properties of~$H(n,d)$ together with known properties of the theta number to reformulate this bound in terms of the eigenvalues of the adjacency matrix of this graph; 
third, we bound the eigenvalues to get the desired result. \\

\noindent{\bf Step 1: The Lov\'asz theta number.}
This parameter
 was introduced by Lov\'asz \cite{Lovasz:1979} to upper bound the Shannon capacity of a graph.} 
 Among its many equivalent definitions, we will use the following primal and dual formulations:
\beq \label{eq:thetadual}
\begin{aligned}
\vartheta(G) = &\max \; \sum_{i,j \in V(G)} X_{ij} \quad \text{s.t.} \quad X \in \mathcal S^{|V(G)|}_{+}, \quad \sum_{i \in V(G)} X_{ii} = 1, \quad  X_{ij} = 0 \quad \forall ij \in E(G), \\
\vartheta(G) = \min \; & t \quad \text{s.t.} \quad X \in \mathcal S^{|V(G)|}_{+}, \quad X_{ii} = t-1 \quad \forall i \in V(G), \quad  X_{ij} = -1 \quad \forall ij \in E(\overline{G}),
\end{aligned}
\eeq
where $\mathcal S^{n}_{+}$ is the set of $n \times n$ symmetric positive semidefinite matrices.

Lov\'asz \cite{Lovasz:1979} proved that the theta number lower bounds the minimum dimension of an orthonormal representation where the vectors are real valued. 
Note that this is slightly different from our setting where we allow the vectors  to have complex entries.
However, we show that the Lov\'asz theta number is  also a lower bound for $\xi(G)$. The proof is an adaptation to the complex case of a known proof~\cite{Monique}.

\begin{lemma}\label{lem:xitheta}
For any graph $G$, we have $\xi(G) \geq \vartheta(\overline{G})$. 
\end{lemma}

\begin{proof}
Let $n = |V(G)|$ and label the vertices of the graph $G$ by $\{1,2,\dots,n\}$.
Suppose that the orthogonal rank of $G$ is equal to $d$ and that $u_1,\dots,u_n \in \mathbb{C}^d$ are the unit vectors forming an orthonormal representation of $G$. 
For every vertex of the graph $i \in [n]$, define a matrix $U_i = u_i u_i^\dagger$ and $U_0 = I_d$.
Let $Z$ be a $(n+1)\times (n+1)$ matrix where the $i,j$-th entry 
$Z_{ij} = \langle U_i, U_j \rangle = \Tr(U_j^\dagger U_i)$ for every $i,j \in  \{0\}  \cup [n]$.
Notice that $Z$ is positive semidefinite since it is the Gram matrix of a set of complex vectors. 
Moreover, $Z$ is real valued and we get that $Z_{00} = d$, $Z_{0i} = \langle I, u_i u_i^\dagger \rangle = 1$ and $Z_{ii} = \langle u_i u_i^\dagger, u_i u_i^\dagger \rangle = (u_i^\dagger u_i)^2 = 1$  for all $i \in V(G)$ and that  $Z_{ij} = (u_i^\dagger u_j) (u_j^\dagger u_i) \ge 0$ for all $i,j \in V(G)$ with equality if $ij \in E(G)$.
By taking the Schur complement\footnote{
Let $X$ be a symmetric matrix of the form
$X=\big(\begin{smallmatrix} \alpha & b^T\cr b & A\end{smallmatrix}\big),$ where  
$b\in \R^{n-1} \text{ and } \alpha >0$. 
$X$ is positive semidefinite if and only if $A- 
bb^T/\alpha$ is positive semidefinite.
The matrix $A-
 bb^T/\alpha$ is called the Schur complement in $X$ with respect to the entry $\alpha$.
} 
in ${Z}$ with respect to the entry ${Z}_{00}$,
we obtain a new symmetric positive semidefinite matrix $X$ with $X_{ii} = 1-1/d$ for all $i \in V(G)$ and $X_{ij} = -1/d$ for all $ij \in E(G)$.
Rescaling $X$ by $d$, we get a feasible solution for the minimization program in (\ref{eq:thetadual}) of $\vartheta(\overline{G})$ with value $d$. We conclude that $d = \xi(G) \geq \vartheta(\overline{G})$.
\end{proof}

\noindent{\bf Step~2: Eigenvalue bound on the theta number.}
In the second step we show that the theta number of the graph~$H(n,d)$ can be expressed in terms of the eigenvalues of its adjacency matrix.
For the remainder of this step, by the eigenvalues of a graph we mean the eigenvalues of its adjacency matrix.

\begin{lemma}\label{lem:Hdn-theta}
For every positive integer~$n$ and~$d\in[n]$, we have
$\vartheta(\overline{H(n,d)}) = 1-\binom{n}{d}/{\lmin}$, where $\lmin$ is the smallest eigenvalue of~$H(n,d)$.
\end{lemma}

Let us recall the following standard definitions.
Let~$G = (V, E)$ be a graph.
A permutation $\pi:V\to V$ is \emph{edge preserving} if for every edge~$\{u,v\}\in E$, we have~$\{\pi(u), \pi(v)\}\in E$.
The graph~$G$ is \emph{vertex-transitive} if for every pair of vertices $u,v\in V$ there is an edge-preserving permutation $\pi:V\to V$ such that~$\pi(u) = v$.
Moreover, $G$ is \emph{edge-transitive} if for every pair of edges $\{u_1,v_1\}, \{u_2,v_2\}\in E$, there is an edge-preserving permutation~$\pi:V\to V$ such that $\pi(u_1) = u_2$ and~$\pi(v_1) = v_2$.
Lov\'asz~\cite[Theorems 8 and 9]{Lovasz:1979} showed that if a graph is both vertex- and edge-transitive, then the theta number is given by a simple formula involving its eigenvalues.

\begin{lemma}[Lov\'asz] \label{cor:theta-transgraph}
For a positive integer~$n$
let~$G$ be an $n$-vertex graph
with eigenvalues
 $\lambda_1 \geq \dots \geq \lambda_n$.
If $G$ is both vertex- and edge-transitive, then $\vartheta(\overline{G}) = 1  -\lambda_1/\lambda_n$.
\end{lemma}

\begin{proof}[Proof of Lemma~\ref{lem:Hdn-theta}]
We start by showing that $H(n,d)$ is vertex-transitive.
Given any pair of vertices $u,v \in \{0,1\}^n$ of $H(n,d)$, consider the automorphism of the graph $H(n,d)$ that maps $x \mapsto x \oplus u \oplus v$ where $\oplus$ is the bit-wise addition.
This map preserves the Hamming distance, and therefore the adjacencies, between the vertices and sends $u \mapsto v$. Hence $H(n,d)$ is vertex-transitive.

To show that $H(n,d)$ is edge-transitive, fix any two edges $uv$ and $st$ and let $p = u \oplus v$,  $q = s \oplus t$. 
Noting that the $n$-bit strings $p$ and $q$ have the same Hamming weight $d$, let $\pi$ be a permutation of the indices such that $\pi(p) = q$.
We define  $\nu$ to be an automorphism that sends a vertex $x$ to $\pi(x \oplus u) \oplus s$.
The map $\nu$ preserves the edges of $H(n,d)$ and, since the permutation $\pi$ maps the all-zero string to itself and $p$ to $q$, we have that $\nu(u)  = s$ and $\nu(v) = t$. Hence, $H(n,d)$ is edge-transitive.

Finally, since the largest eigenvalue of a vertex-transitive graph is equal to its degree, we clearly have~$\lambda_1(H(n,d)) = \binom{n}{d}$.
The result now follows from Lemma~\ref{cor:theta-transgraph}.
\end{proof}

\noindent{\bf Step 3: Bound on the smallest eigenvalue of $H(n,d)$.}
Finally, we prove an upper bound on the magnitude of the smallest eigenvalue of~$H(n,d)$.

\begin{lemma}\label{lem:small-eig} 
Let~$n$ and~$d$ be even positive integers such that~$d<n/2$.
Then, the smallest eigenvalue $\lmin$ of the graph $H(n,d)$ is a negative number such that 
\beqn
|\lmin| \leq \sqrt{\frac{2^n \binom{n}{d}}{\binom{n}{n/2 - \sqrt{d(n-d)}}}}.
\eeqn
\end{lemma}

The proof of the lemma uses the following facts from coding theory that can be found in the survey~\cite{Delsarte:1998}.
The eigenvalues of $H(n,d)$ play a fundamental role
in the theory of Hamming association schemes, 
where they are expressed in terms of a set of orthogonal polynomials known as the (binary) \emph{Krawtchouk polynomials}.
For a positive integer~$n$ and~$d = 0,1,\dots,n$, the 
Krawtchouk polynomial~$K_d^n \in \R[x]$ is a degree-$d$ polynomial that is uniquely defined by
 \beq\label{Krawtchouk}
 K_d^n(x) = \sum_{j=0}^d (-1)^j \binom{x}{j} \binom{n-x}{d-j},
 \quad\quad
 x = 0,1,\dots,n.
 \eeq
When~$n$ and~$d$ are even, then $K_d^n$ is symmetric about the point~$x = n/2$.
Moreover, these polynomials satisfy the important orthogonality relation
\beq\label{eq:kraw-orth}
\sum_{x=0}^n \binom{n}{x} K_{d}^n(x)K_{d'}^n(x)
=
\delta_{d,d'} \binom{n}{d} 2^n.
\eeq
The set of distinct eigenvalues of $H(n,d)$ turns out to be $\{K_d^n(0), K_d^n(1), \dots, K_d^n(n)\}$. 
Crucial to our proof of Lemma~\ref{lem:small-eig} then, is the following result of
Levenshtein~\cite[Theorem 6.1]{Levenshtein:1995} characterizing the smallest roots of the Krawtchouk polynomials.

\begin{theorem}[Levenshtein]\label{thm:lev}
Let~$n$ be a positive integer and~$d\in[n]$.
Then,~$K_d^n$ has exactly~$d$ distinct roots and its smallest root is given by
\beq\label{eq:lev-bound}
n/2 - \max_{z} \Big( \sum_{i=0}^{d-2}z_i z_{i+1} \sqrt{(i+1)(n-i)} \Big),
\eeq
where the maximum is over all vectors $z = (z_0,\dots,z_{d-1})$ on the real Euclidean unit sphere.
\end{theorem}

This implies the following general bound on the location of the smallest root of~$K_d^n$.
The bound is stated for instance in~\cite{Krasikov:2001}, but since we were unable to find a published proof we include one here for completeness.

\begin{corollary}\label{cor:kraw-roughbound}
Let~$n$ and~$d$ be positive integers such that~$d<n/2$.
Then, the smallest root of~$K_d^n$ lies in the interval
$
\big[n/2 - \sqrt{(n-d)d}, n/2\big]
$.
\end{corollary}

\begin{proof}
It is clear that~\eqref{eq:lev-bound} is trivially upper bounded by $n/2$. 
We focus on the lower bound. 
To this end, let~$z = (z_0,\dots,z_{d-1})$ be a real unit vector achieving the maximum in~\eqref{eq:lev-bound}.
For $i\in\{0,1,\dots,d-1\}$ define the numbers $a_i = z_i\sqrt{n-i}$ and $b_i = z_{i+1} \sqrt{i+1}$.
By the Cauchy-Schwarz inequality,
\begin{align}
\Big( \sum_{i=0}^{d-2}z_i z_{i+1} \sqrt{(i+1)(n-i)} \Big)^2 
&
=
\Big(\sum_{i=0}^{d-2} a_i b_i\Big)^2 \nonumber \\
&\leq
\Big(\sum_{i=0}^{d-2} a_i^2\Big)
\Big(\sum_{j=0}^{d-2} b_j^2\Big) \nonumber\\
&= 
\Big(\sum_{i=0}^{d-2} a_i^2\Big)
\Big(\sum_{j=1}^{d-1} b_{j-1}^2\Big)\nonumber\\
&\leq 
\Big(\sum_{i=0}^{d-1} a_i^2\Big)
\Big(\sum_{j=0}^{d-1} b_{j-1}^2\Big)\nonumber\\
&=
\Big(\sum_{i=0}^{d-1}z_i^2(n-i)\Big)
\Big(\sum_{j=0}^{d-1} z_{j}^2 j \Big) \nonumber\\
&=
\Big( n - \sum_{i=0}^{d-1} z_{i}^2 i  \Big)
\Big( \sum_{j=0}^{d-1}  z_{j}^2 j \Big), \label{eq:lev-rhs}
\end{align}
where in the last equality we used the fact that~$z$ is a unit vector.
Observe that the sum
$
\sum_{i=0}^{d-1}  z_{i}^2 i 
$
lies in the interval~$[0,d-1]$.
Hence, since $d<n/2$,~\eqref{eq:lev-rhs} is at most~$\max\{(n - t)t\st t\in[0,d-1]\} = (n-(d-1))(d-1) \leq (n-d)d$.
\end{proof}

\begin{proof}[Proof of Lemma~\ref{lem:small-eig}]
Since the trace of a matrix equals the sum of its eigenvalues and the trace of an adjacency matrix is zero, it follows that
 $\lmin < 0$. 
 
 Recall that the eigenvalues of~$H(n,d)$ belong to the set~$\{K_d^n(x)\st x=0,1,\dots,n\}$.
Moreover, since by our assumption $n$ and $d$ are even, the polynomial~$K_d^n$ is symmetric about the point~$n/2$.
Also observe that~$K_d^n(0) > 0$ and hence the first time this polynomial assumes a negative value is somewhere beyond its smallest root, i.e. the smallest $x$ for which $K_d^n(x)<0$ lies in between the smallest root and $n/2$.
It therefore follows from Corollary~\ref{cor:kraw-roughbound} and from the fact that~$K_d^n$ is symmetric about the point~$n/2$  that~$\lmin = K_d^n(x^\star)$ for some integer~$x^\star\in [n/2- \sqrt{(n-d)d}, n/2]$.

Clearly~\eqref{eq:kraw-orth} implies that
\beqn
\sum_{x=0}^n \binom{n}{x} K_d^n(x)^2 = \binom{n}{d} 2^n.
\eeqn
Hence,
\beqn
\binom{n}{x^{\star}} K_d^n(x^{\star})^2 \leq \binom{n}{d} 2^n
\eeqn
and we can conclude that
\beqn
|\lmin|^2 
= 
|K_d^n(x^\star)|^2 
\leq  
\frac{2^n \binom{n}{d} }{ \binom{n}{x^\star} }
\leq
\frac{2^n \binom{n}{d} }{ \binom{n}{n/2 - \sqrt{(n-d)d} }}.
\eeqn
\end{proof}

\noindent{\bf Putting everything together.}
To conclude this section, we combine the main lemmas of the above three steps to prove Theorem~\ref{thm:xi-smaller12}.

\begin{proof}[Proof of Theorem~\ref{thm:xi-smaller12}]
Combining Lemmas~\ref{lem:xitheta},~\ref{lem:Hdn-theta}, and~\ref{lem:small-eig} gives
\beq\label{eq:fluffy}
\xi(H(n,d)) 
\geq 
\vartheta\big(\overline{H(n,d)}\big)
\geq
 1- \binom{n}{d}/\lmin
\geq 
1 + \sqrt{\frac{\binom{n}{d}\binom{n}{n/2 - \sqrt{(n-d)d}}}{2^n}}.
\eeq
We take the logarithm and use Stirling's approximation: $\log \binom{n}{k} = \big(H(k/n) + o(1)\big)n$, where $H(p) = - p \log p - (1-p)\log(1-p)$ is the binary entropy function and the~$o(1)$ term goes to zero as~$n\to\infty$ (see for example~\cite[p.~64]{Spencer:2014}).
Then, for $\alpha = d/n$, the logarithm of~\eqref{eq:fluffy} is at least
\beqn
\frac{1}{2} \log \left( \frac{\binom{n}{d}\binom{n}{n/2 - \sqrt{(n-d)d}}}{2^n} \right)
= \frac{n}{2} \left( H(\alpha) + H\left( 1/2 - \sqrt{(1-\alpha)\alpha} \right) -1 + o(1) \right).
\eeqn
Using properties of the binary entropy, it can be shown that
 $H(\alpha) + H ( 1/2 - \sqrt{(1-\alpha)\alpha} ) - 1 > 0$ for any $\alpha \in (0, 1/2)$ (for the proof see Appendix~\ref{app:entropy}). 
 In particular,  $\log \xi(H(n,n/4)) \geq \Omega(n)$.
\end{proof}

\subsubsection{Two-round quantum communication}\label{sec:2qc}
%
%
Using a distributed version of Grover's search algorithm, we find a quantum  protocol that solves \peq-$\binom{n}{n/4}$
 with a logarithmic number of qubits, which gives the second part of Theorem~\ref{thm:promise-general}.
%

\begin{theorem}\label{thm:exact-grover-n/4}
The two-round quantum communication complexity of \peq-$\binom{n}{n/4}$ is at most $ 2 \ceil{\log n} + 1$ qubits.
\end{theorem}

\begin{proof}
Let $x$ and $y$ be the inputs of Alice and Bob, respectively, and $z = x \oplus y$ be their bit-wise addition.
The promise ensures that either $|z| = 0$ if $x = y$ or $|z| = n/4$ in the case where $x \neq y$.

If a bit string $z\in \{0,1\}^n$ is known to contain exactly $n/4$ entries that are 1,
Grover's algorithm~\cite{Grover:1996} is able to find one of these entries without error~\cite{Boyer:1998}, needing only a single query to
the string $z$. For any string we define the query unitary $U_z=\sum_{i=1}^{n} (-1)^{z_i}  \ketbra{i}{i}$ and we define
$\ket{s} = \frac{1}{\sqrt{n}} \sum_{i=1}^{n} \ket{i}$ to be the uniform superposition of all basis states. 
 Then $G = 2\ketbra{s}{s} - I$ is a unitary operation known as the Grover diffusion operator. 

The quantum communication protocol can be viewed as combining Grover's algorithm with a special case of the simulation theorem given in~\cite[Theorem~2.1]{Buhrman:1998}.
We want to perform the algorithm on
the effective string $z = x \oplus y$, using the fact that performing a single query $U_z$ is the same 
as performing the operations $U_x$ and $U_y$ in sequence, \ie, $U_z = U_x U_y = U_y U_x$. 

At the start of the protocol, Bob  creates the state $U_y \ket{s}$ and sends this state over to Alice using $\ceil{\log {n}}$ qubits.
Alice first applies $U_x$ to the incoming state and then applies
the Grover operator $G$. The final state 
of Grover's algorithm is
$\frac{1}{\sqrt{n/4}} \sum_{i \text{ s.t.\ } z_i=1} \ket{i}$ if $|z|=n/4$. 
That is, in the case that $x \neq y$, Grover's algorithm has
produced a superposition over all indices $i$ such that $x_i \neq y_i$.
Alice measures the state, obtaining
some index $i^*$ such that $x_{i^*} \neq y_{i^*}$ if $x \neq y$. 
Then she sends $i^*$ and the value $x_{i^*}$ over to Bob using $\ceil{\log n}+1$ qubits. He outputs `equal' if and only if $x_{i^*} = y_{i^*}$. The total communication cost
of the protocol is then $2 \ceil{\log n} + 1$ qubits.
\end{proof}

The above protocol can be extended to efficiently solve \peq-$\binom{n}{\alpha n}$ for constants~$\alpha < 1/2$ (independent of~$n$) in a constant number of rounds, by using a more general exact version of the Grover search algorithm. This construction is described in Appendix~\ref{sec:exact-grover-general}.

\subsection{Distances close to~\texorpdfstring{$n/2$}{n/2}}

In~\cite{Buhrman:1998, Gruska:2014} it is shown that for any~$\alpha \geq 1/2$, the problem \peq-$\binom{n}{\alpha n}$ admits an $O(\log n)$-qubit one-round quantum protocol, whereas the proof of our Theorem~\ref{thm:xi-smaller12} shows that for~$\alpha \in (0,1/2)$ (independent of~$n$), the one-round quantum communication complexity of \peq-$\binom{n}{\alpha n}$ is at least~$\Omega(n)$.
Does the threshold of this exponential jump sit exactly at $1/2$? In the following simple lemma, we prove that this is not the case. When  $\alpha$ is strictly smaller than $1/2$ but very close to it, the one-round quantum communication complexity of \peq-$\binom{n}{\alpha n}$ still requires only a logarithmic number of qubits.

\begin{lemma}
For $d = n/2 - \ell$ with $\ell \leq O(\log n)$, $n$ and $d$ even,
the one-round quantum communication complexity of \peq-$\binom{n}{d}$ is at most $O(\log n)$.
\end{lemma}

\begin{proof}
We first make the following easy observation. Suppose Alice sends to Bob the first $2\ell$ bits of her input. If this $2 \ell$-bit string differ from Bob's initial part of the input, he knows that the answer is `not equal'. Otherwise Alice and Bob have to exchange information about the remaining part of their inputs, which have length $n' = n - 2\ell$ and they are either equal or differ in exactly $d' = d = n/2 - \ell = n'/2$ positions.

Hence, the map $\phi : \{0,1\}^n \to \mathbb{C}^{2^{2\ell} n'}$ that sends $x \mapsto x_1 \otimes x_2 \otimes \dots \otimes x_{2 \ell} \otimes \frac{1}{\sqrt{n'}} \sum_{i=1}^{n'} (-1)^{x_{i+2\ell}} e_i$, where $e_i$ is the $i$-th canonical basis vector, is an orthonormal representation of the graph $H(n,d)$.
The result now follows from Theorem~\ref{thm:1roundquantum}.
\end{proof}

\section{The list problem}\label{sec:list}
In this section, we consider the~{\sc $\mathcal L$-list} problem:  Bob gets a list~$L\in \mathcal L$ from a family~$\mathcal L\subseteq  2^{\mathcal X}$ of lists, Alice gets an element~$x\in L$, and Bob must learn~$x$.
\subsection{Classical communication complexity of the list problem}

Here we prove Theorem~\ref{thm:list-general}.
The list problem that gives the result is simple: For positive integers~$k,N$ such that~$2\leq k\leq N$, we consider the list problem~$\mathcal L = \binom{\mathcal [N]}{k}$, where the family of lists consists of all~$k$-element subsets of~$[N]$.
Note that for this~$\mathcal L$, we clearly have~$\omega(\mathcal L) = k$ (not to be confused with $\omega(G_{\mathcal L}) = N$) and that~$G_{\mathcal L}$ is the complete graph on~$N$ vertices, giving~$\chi(G_\mathcal L) = N$.
Hence, Theorem~\ref{thm:4-rounds} gives a four-round protocol using at most $\log \log N + 2 \log k + O(\log \log k)$ bits of communication.

\begin{theorem}\label{thm:list}
The classical communication complexity of {\sc $\binom{[N]}{k}$-list} is at least 
\beqn
\log \log N + 2 \log (k-1) - \log \log (k-1) - O(1).
\eeqn
\end{theorem}
 
To see that this implies Theorem~\ref{thm:list-general} note that the above bound 
 can be written as
$\log\log\chi(G_\mathcal L) + (2-o(1))\log\omega(\mathcal L)$, where the term~$o(1)$ goes to zero as~$k\to\infty$.
Choosing~$k = \log N$ then gives the second part of the theorem.

To prove Theorem~\ref{thm:list}, we use a bound on the size of cover-free families due to D\'yachkov and Rykov~\cite{Dyachkov:1982}; see~\cite{Ruszinko:1994, Furedi:1996} for simplified proofs (in English).

\begin{definition}
Let~$r$ be a positive integer and~$\mathcal S$ be a finite set.
A family $\mathcal F\subseteq 2^{\mathcal S}$ of at least~$r+1$ subsets is \emph{$r$-cover-free} if  every subfamily of $r+1$ distinct sets $F_0,F_1,\dots,F_r \in \mathcal F$ satisfies
$F_0 \nsubseteq F_1 \cup \dots \cup F_r$. 
\end{definition}

\begin{theorem}[D\'yachkov--Rykov]\label{thm:coverfree}
There exists an absolute constant~$c>0$ such that the following holds.
Let~$N$ and~$r$ be positive integers such that~$N \geq r+1$ and~$r\geq 2$. Let~$\mathcal S$ be a finite set.
Let $\mathcal F\subseteq 2^{\mathcal S}$ be an $r$-cover free family consisting of~$N$ sets.
Then, 
\beqn
|\mathcal S| \geq \frac{cr^2 \log N}{\log r}.
\eeqn
\end{theorem}

\begin{proof}[Proof of Theorem~\ref{thm:list}]
For a positive integer~$C$, suppose that the communication complexity of {\sc $\binom{[N]}{k}$-list} is~$C$.
Fix such a protocol and for every input pair~$(x,L)$ in the {\sc $\binom{[N]}{k}$-list} problem, define the \emph{transcript} $T_{x,L}\in \bset{}\cup\bset{2}\cup\cdots\cup\bset{C}$ as the concatenation of the parties' messages in the order they are sent during their conversation on input~$(x,L)$.
Let~$\mathcal T$ be the set of said transcripts.

For each transcript~$T\in\mathcal T$, denote by~$T^A$ the sequence of Alice's messages in~$T$, to be understood as a sequence of strings indexed by her rounds in the conversation.
Let~$\mathcal F  = \{ F_{x} \}_{x \in \mathcal X} \subseteq 2^{\mathcal T}$ be the family where each $F_{x}$ is the collection of transcripts~$T\in \mathcal T$ that is consistent with $x$ being Alice's input and that  agrees on~$T^A$.
We claim that~$\mathcal F$ is a $(k-1)$-cover free family.
To see this, take any $k$ sets of $\mathcal F$, say~$F_{x_0},\dots,F_{x_{k-1}}$, and let $L$ be the corresponding $k$-element list $ \{x_0,\dots,x_{k-1}\}$.
Consider the transcript $T_{x_{0},L}$ of the input pair $(x_{0},L)$.
Clearly,~$T_{x_0,L}\in F_{x_0}$.
We show that~$T_{x_0,L}\not \in F_{x_i}$ for each ${i\in\{1,\dots,k-1\}}$, which gives the claim as this implies that~$F_{x_0}\not\subseteq F_{x_1}\cup\cdots\cup F_{x_{k-1}}$.
Suppose that $T_{x_0,L} \in F_{x_i}$ holds for some ${i\in\{1,\dots,k-1\}}$.
This means that Alice sends identical message sequences 
on inputs~$x_0$ and~$x_i$ and therefore that Bob is not able to distinguish between these two cases for the input pair $(x_{0},L)$,
  contradicting our assumption that we started with a functional protocol.

We also claim that~$\mathcal F$ consists of at least~$N$ sets.
Indeed, for every pair ${x,y\in[N]}$, there is a list~$L\in\binom{[N]}{k}$ containing both~$x$ and~$y$.
Since we must have~$T_{x,L}^A \ne T_{y,L}^A$ for Bob to be able to distinguish~$x$ and~$y$ on input~$L$, the inputs~$x$ and~$y$ induce distinct transcript sets.

It thus follows from Theorem~\ref{thm:coverfree} that the total number of distinct transcripts is at least
\beqn
|\mathcal T| \geq \frac{c(k-1)^2\log N}{\log (k-1)}.
\eeqn
Hence, since~$\mathcal T\subseteq \bset{}\cup\bset{2}\cup\cdots\cup\bset{C}$, we have
\beqn
\frac{2^{C+1} - 1}{2-1}
=
\sum_{l=0}^C2^l
\geq
\frac{c(k-1)^2\log N}{\log (k-1)},
\eeqn
for some absolute constant~$c>0$.
Taking logarithms now gives the claim.
\end{proof}

\subsection{Quantum communication complexity of the list problem}

Analogous to Witsenhausen's result, the one-round quantum communication complexity of a list problem is characterized in terms of the orthogonality dimension of its associated graph.

\begin{lemma}\label{lem:list-1-quantum}
For every family~$\mathcal L\subseteq 2^{\mathcal X}$, the one-round quantum communication complexity of {\sc $\mathcal L$-list} equals~$\ceil{\log \xi(G_{\mathcal L})}$.
\end{lemma}

\begin{proof}
Consider an optimal one-round protocol. Without loss of generality, we can assume that Alice sends to Bob a pure state~$\ket{\phi_x} \in\C^d$ on input $x\in\mathcal X$.
Then, given a  list~$L\in\mathcal L$, Bob has a measurement that allows him to distinguish the states~$\{\ket{\phi_x}\st x\in L\}$.
It thus follows from Lemma~\ref{lem:orth} that these states must be orthogonal.
In particular, since for every list~$L\in\mathcal L$, each pair of distinct elements~$x,y\in L$ forms an edge in~$G_{\mathcal L}$, the vectors~$\ket{\phi_x}$, $x\in\mathcal X$, form a $d$-dimensional orthogonal representation.
Hence, $\xi(G_\mathcal L) \leq d$.

Conversely, let $f:V(G_\mathcal L)\to\C^d$ be an orthogonal representation of $G_\mathcal L$. 
Then, for every list~$L\in\mathcal L$, the vectors~$\{f(x)\st x\in L\}$ are pairwise orthogonal.
If Bob gets a list~$L\in\mathcal L$ and Alice gets an element~$x\in L$, it follows from  Lemma~\ref{lem:orth} that there is a quantum measurement allowing Bob to uniquely identify~$x$ when Alice sends~$f(x)$ using $\log d$-qubits.
Hence, the one-round quantum communication complexity is at most~$\ceil{\log \xi(G_\mathcal L)}$.
\end{proof}

For multi-round protocols, a quantum analogue of Theorem~\ref{thm:lb-classical} also holds.

\begin{lemma}\label{lem:lb-quantum}
For every family~$\mathcal L\subseteq 2^{\mathcal X}$, the quantum communication complexity of {\sc $\mathcal L$-list} is at least $\max\{ \Omega(\log \log \chi(G_{\mathcal L})), \log \omega(\mathcal L) \}$.
\end{lemma}

\begin{proof}
Kremer's Theorem (Theorem~\ref{thm:kremer}) shows that there is at most an exponential difference between the (multi-round) quantum and one-round classical communication complexity.
Hence, by Witsenhausen's result, the former is at least~$\Omega(\log\log \chi(G_{\mathcal L}))$.
Moreover, on the worst input Bob has to be able to distinguish among~$\omega(\mathcal L)$ different elements. Hence, $\log \omega(\mathcal L)$ bits of information must be communicated and Holevo's Theorem~\cite{Holevo:1973} says that to retrieve $\log \omega(\mathcal L)$ bits of information $\log \omega(\mathcal L)$ qubits are necessary.
\end{proof}

\subsection{Proof of Theorem~\ref{thm:simple2rnd}}

Recall that we are considering the following family of lists. 
For an even positive integer~$n$ and~$d\in[n]$, let~$\mathcal L_d \subseteq 2^{\bset{n}}$ be the family of all lists~$L\subseteq \bset{n}$ of maximal cardinality such that all strings in~$L$ have Hamming distance exactly~$d$.
We denote by~$\mathcal K$ the union $\mathcal L_{n/2}\cup\cdots\cup\mathcal L_{n}$.
In other words, this is the union of lists for which, individually, there is a one-round $O(\log n)$-qubit protocol.  

\begin{proof}[Proof of Theorem~\ref{thm:simple2rnd}]
Let $\ell = \ceil{\log n}$ and $U$ be the~$(\ell+1)$-qubit unitary matrix which satisfies
$U  \ket{0}  \ket{0}^{\otimes \ell} = \ket{0} \ket{0}^{\otimes \ell}$ and $U \ket{1} \ket{0}^{\otimes \ell}=\frac{1}{\sqrt{n}}\ket{1} \sum_{i=1}^{n}\ket{i}$.
Moreover, for any $2^\ell$-bit string $z$, we define the conditional query unitary $U_z$ which acts on the computational basis states as 
$U_z \ket{0} \ket{i} = \ket{0} \ket{i}$ and $U_z \ket{1} \ket{i} = (-1)^{z_{i}} \ket{1} \ket{i}$ for any $i \in [2^{\ell}]$.
For a small technicality, 
if $n$ is not a power of $2$, i.e. $\ell>\log n$, we will map any $n$-bit string to a $2^{\ell}$-bit string obtained by padding zeros to the original string. 
We can now explain the protocol.

Consider the input pair $(x,L)$ where~$L \in \mathcal L_{d}$. 
Bob looking at the list $L$ learns $d$ and sends to Alice the single qubit $\gamma \ket{0} + \sqrt{1-\gamma^2} \ket{1}$ where $\gamma^2 = 1 - \frac{n}{2d} \geq 0$. 
Alice pads the state $\ket{0}^{\otimes \ell}$ to the one she received and then applies in sequence the unitaries $U$ and $U_{x}$, obtaining the state 
$\ket{\phi_x} = U_{x} U \left( (\gamma \ket{0} + \sqrt{1-\gamma^2} \ket{1})\ket{0}^{\otimes \ell}\right) = \gamma \ket{0} \ket{0}^{\otimes \ell} + \sqrt{\frac{{1-\gamma^2}}{{n}}} \sum_{i=1}^{n} (-1)^{x_{i}} \ket{1}\ket{i}$.
She sends this to Bob using $\ceil{\log n} + 1$ qubits.
Notice that if $x, y \in \{0,1\}^{n}$ differ in exactly $d$ positions, then the states $\ket{\phi_x}$ and $\ket{\phi_y}$ are orthogonal to each other. 
Hence, by Lemma~\ref{lem:orth}, using the list $L$ Bob can perform a measurement that allows him to learn Alice's input $x$.
\end{proof}

\subsection{Entanglement-assisted and non-signaling communication complexity of the list problem}

In this last section, we present two results regarding the list problem when the players can only exchange classical bits but they are allow to share non-classical correlations.  

\paragraph{Quantum correlations.}
If Alice and Bob can share an entangled state and communicate classical bits, then they can use the teleportation protocol of Bennett et al.~\cite{Bennett:1993} to simulate the quantum communication with a factor of 2 overhead. 
However, there may be more efficient protocols. 
In particular, we show for the {\sc $\mathcal K$-list} problem, where~$\mathcal K$ is the union $\mathcal L_{n/2}\cup\cdots\cup\mathcal L_{n}$, a two-round entanglement-assisted protocol that uses only $\ceil{\log n} + 3$ classical bits.

\begin{lemma}\label{lem:entang}
For $\mathcal K = \mathcal L_{n/2}\cup\cdots\cup\mathcal L_{n}$, there exists a two-round entanglement-assisted protocol for {\sc $\mathcal K$-list} 
where Bob sends Alice a single bit and Alice replies with $\ceil{\log n}+2$ bits of communication.
\end{lemma}

\begin{proof}
Let $(x,L)$ be Alice and Bob's input pair where $L \in \mathcal L_{d}$.
Consider the conditional query unitary $U_z$, where $z$ is a $n$-bit string, which acts on the computational basis states as 
$U_z \ket{0} \ket{i} = \ket{0} \ket{i}$ and $U_z \ket{1} \ket{i} = (-1)^{z_{i}} \ket{1} \ket{i}$ for any $i \in \{0,1\}^{n}$.
We show a two-round communication protocol that uses as shared entanglement the state $\frac{1}{\sqrt{n}} \sum_{i=1}^{n} \ket{i}_A \ket{i}_B$ together with two EPR pairs. 
We use the subscript $A$ (resp.\ $B$) to specify Alice's part of the state (resp.\ Bob's).

From the list $L$, Bob learns the distance $d$ and uses one EPR pair and one bit of communication to remote state prepare the qubit $\gamma \ket{0} + \sqrt{1-\gamma^2} \ket{1}$, where $\gamma^2 = 1 - \frac{n}{2d} \geq 0$ \cite{Pati2000,Lo2000}. 
Now Alice and Bob are sharing the entangled state: $\left( \gamma \ket{0}_A + \sqrt{1-\gamma^2} \ket{1}_A \right) \frac{1}{\sqrt{n}} \sum_{i=1}^{n} \ket{i}_A \ket{i}_B$.
Using her input $x$, Alice performs the unitary $U_{x}$ followed by the unitary $U  = \ketbra{0}{0} \otimes F_n + \ketbra{1}{1} \otimes F_n$ where $F_n$ is the $n \times n$ discrete quantum Fourier transform.
The entangled state is now:
\beqn
\gamma \ket{0}_A  \frac{1}{n} \sum_{i=1}^{n} \sum_{j =1}^{n} \omega^{ij}_n \ket{j}_A \ket{i}_B
+ 
(\sqrt{1-\gamma^2}) \ket{1}_A \frac{1}{n} \sum_{i=1}^{n} \sum_{j=1}^{n} (-1)^{x_{i}} \omega^{ij}_n \ket{j}_A \ket{i}_B,
\eeqn 
where $\omega_n$ is the $n$-th root of unity.
Alice measures in the computational basis her second $n$ qubits register and gets an outcome $\hat{j}$.
She sends $\hat{j}$ to Bob using $\ceil{\log n}$ classical bits.
Moreover, Alice teleports the qubit  $\gamma \ket{0} + \sqrt{1-\gamma^2} \ket{1}$  to Bob using the protocol of Bennett et al.~\cite{Bennett:1993}. This requires two classical bits and an EPR pair. 

He can then use $\hat{j}$ to perform a unitary $(\sum_i \omega_n^{i*} \ketbra{i}{i})^{\hat{j}}$ that will put his register in the state 
\beqn
\ket{\psi_x}= \gamma \ket{0} \frac{1}{\sqrt{n}} \sum_{i=1}^{n} \ket{i}
+ 
(\sqrt{1-\gamma^2}) \ket{1}  \frac{1}{\sqrt{n}} \sum_{i=1}^{n} (-1)^{x_{i}} \ket{i}.
\eeqn 
At last, we notice that if $x,y \in \{0,1\}^{n}$ differ in exactly $d$ bits then $\ket{\psi_x}$ is orthogonal to $\ket{\psi_y}$.
Using the elements of the list $L$, by Lemma~\ref{lem:orth}, Bob can construct a measurement that allows his to learn Alice's input. 
In total the protocol required $\ceil{\log n} + 2$ EPR pairs and $\ceil{\log n} + 3$ bits of classical communication.
\end{proof}

\paragraph{Non-signaling correlations.}
If the two parties can share non-signaling correlations, every list problem becomes trivial.
Let~$\mathcal{A,B,S}$, and~$\mathcal T$ be some sets.
Alice and Bob are said to \emph{share non-signalling correlations} if there is a (magically correlated) pair of devices taking inputs from~$\mathcal S$ and~$\mathcal T$, respectively, such that if Alice gives~$s\in \mathcal S$ to her device and Bob gives~$t\in\mathcal T$ to his, the devices return elements~$a\in \mathcal A$ and~$b\in \mathcal B$ with probability~$P(a,b|s,t)$, where the joint distributions~$P(\cdot,\cdot|s,t)$ on~$\mathcal A\times\mathcal B$ have marginals satisfying
\begin{align*}
\sum_{b\in\mathcal B}P(a,b|s,t) &= \sum_{b\in\mathcal B}P(a,b|s,t'),\quad \forall t,t'\in\mathcal T\\
\sum_{a\in\mathcal A}P(a,b|s,t) &= \sum_{a\in\mathcal A}P(a,b|s',t),\quad \forall s,s'\in\mathcal S.
\end{align*}
In words, the parties' marginal distributions depend only on their own inputs and not on the other's.

\begin{lemma}
For every family~$\mathcal L\subseteq 2^{\mathcal X}$,  there is a one-round non-signaling protocol that uses only $\ceil{\log \omega(\mathcal L)}$ bits of communication and this is optimal.
\end{lemma}

\begin{proof}
Fix a list~$L\in \mathcal L$ and an element~$x\in L$.
Uniquely label the elements of~$L$ with numbers in~$\Z_{\omega(\mathcal L)}$ and let~$i$ be the label assigned to~$x$.
Let~$P(\cdot,\cdot|x, L)$ be the probability distribution over~$\Z_{\omega(\mathcal L)}\times \Z_{\omega(\mathcal L)}$ that assigns probability~$1/\omega(\mathcal L)$ to each pair in $\{(a, a+i)\st a\in \Z_{\omega(\mathcal L)}\}$ and vanishes on all other pairs.
Clearly this distribution is non-signaling.
Similarly define non-signalling distributions for every other pairs~$(x',L')$ in the list problem.

Now consider the following protocol.
Upon receiving~$x\in \mathcal X$ and~$L\in\mathcal L$ such that~$x\in L$, respectively, Alice and Bob sample from the distribution~$P(\cdot,\cdot|x, L)$ as explained above and get~$a$ and $a+i\in\Z_{\omega(\mathcal L)}$, respectively.
Next, Alice sends~$a$ to Bob, using at most~$\ceil{\log \omega(\mathcal L)}$ bits of communication. 
Finally, Bob subtracts Alice's message from his input, getting ${(a+i) - a = i}$, which tells him Alice's input.

At last, we notice that any functional protocol has to communicate at least~$\ceil{\log \omega(\mathcal L)}$ bits and hence the above protocol is an optimal one. Indeed, there is an instance of the problem where Bob has to distinguish Alice's input from a list of $\omega(\mathcal L)$ different elements. 
\end{proof}

\section*{Acknowledgements}
Part of this work was conducted while H.~B., D.~L., T.~P.\ and 
F.~S.\ were at the Newton Institute of Mathematical Sciences in Cambridge, UK.
J.~B.\ was supported by a Rubicon grant from the Netherlands Organization for Scientific Research (NWO).
H.~B., T.~P.\ and F.~S.\ were funded in part by the EU project SIQS.
D.~L.\ is supported by NSERC, CRC, and CIFAR.

The authors thank Jeroen Zuiddam for useful discussions regarding Section~\ref{sec:equality}. 
T.P. thanks Monique Laurent for providing the proof of Lemma~\ref{lem:xitheta}.
\bibliographystyle{alpha}
\bibliography{RoundChromatic}

\appendix
\renewcommand{\thetheorem}{\Alph{section}\@thmcountersep\arabic{theorem}}

\section{Orthogonal rank is NP-hard}\label{app:hardrank}

\begin{theorem}\label{thm:orthrank}
For any $k \ge 3$, it is NP-hard to decide whether or not $\xi(G) \le k$.
\end{theorem}

Peeters~\cite{Peeters:1996} proved this result in a more general context for the case of $k=3$. We use the same ideas, sometimes simplifying them for our specific case, and then use a trick from~\cite{Laurent:2013} to extend the result to any $k \ge 3$.

\begin{lemma}\cite{Peeters:1996}\label{lem:NP-xi=3}
It is NP-hard to decide whether or not $\xi(G) \le 3$.
\end{lemma}

\begin{proof}
The key idea is, given a graph $G$, to construct in polynomial time a new graph $G'$ with the property that $\chi(G) \le 3$ if and only if $\chi(G') \le 3$, and furthermore $\chi(G') \le 3$ if and only if $\xi(G') \le 3$.
Since deciding whether a graph is three colorable is NP-complete, this implies that determining if $\xi(G) \le 3$ is NP-hard.

Let $G$ be a graph with vertex set $V(G) = [n]$.
The new graph $G'$ is obtained from $G$ by adding for each pair of distinct vertices $i,j \in [n]$ 
the gadget graph $H_{ij}$, see Figure~\ref{fig:gadget}.
Hence, $G'$ has $|V(G)| + 2n(n-1)$ vertices and $|E(G)| + \frac{9}{2}n(n-1)$ edges.
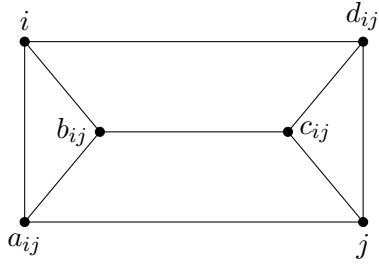
\begin{figure}[t]
\begin{center}

\begin{tikzpicture} [label distance=0.2mm]

\coordinate [label=above:$i$] (i) at (1,3);
\draw [fill=black] (i) circle (0.06cm);
\coordinate [label=below:$a_{ij}$] (a) at (1,0.6);
\draw [fill=black] (a) circle (0.06cm);
\coordinate [label=left:$b_{ij}$] (b) at (2,1.8);
\draw [fill=black] (b) circle (0.06cm);

\coordinate [label=above:$d_{ij}$] (d) at (5.5,3);
\draw [fill=black] (d) circle (0.06cm);
\coordinate [label=below:$j$] (j) at (5.5,0.6);
\draw [fill=black] (j) circle (0.06cm);
\coordinate [label=right:$c_{ij}$] (c) at (4.5,1.8);
\draw [fill=black] (c) circle (0.06cm);

\draw (i) to (a);
\draw (b) to (a);
\draw (i) to (b);
\draw (j) to (c);
\draw (d) to (c);
\draw (j) to (d);
\draw (i) to (d);
\draw (b) to (c);
\draw (a) to (j);

\end{tikzpicture}
 \caption{The gadget graph $H_{ij}$.
}
\label{fig:gadget}
\end{center}
\end{figure}
Clearly, if $\chi(G') =3$ then $\chi(G) \le 3$.
Moreover, if $\chi(G) \le 3$ we can three color the graph $G'$ as following: 
if $i$ and $j$ have the same color then use the color classes $\{\{i,j\},\{a_{ij},c_{ij}\},\{b_{ij},d_{ij}\}\}$, otherwise use $\{\{i,c_{ij}\},\{a_{ij},d_{ij}\},\{b_{ij},j\}\}$.
Thus $\chi(G) \le 3$ if and only if $\chi(G') =3$.

Now, since $3 \le \omega(G') \le \xi(G') \le \chi(G')$ we have that $\chi(G') = 3$ implies $\xi(G') = 3$.
The only thing we have to prove is that if $\xi(G') = 3$ then $\chi(G') = 3$.
Let $\phi:V(G') \to \mathbb{C}^3$ be an orthonormal representation of $G'$.
For every vertex $i$ in the subgraph $G$ of $G'$, let $V_i$ be the 1-dimensional subspace spanned by $\phi(i)$.
We claim that for every $i,j \in [n]$ either $V_{i} = V_{j}$ or $V_{i}$ is orthogonal to $V_{j}$.
Before proving this simple statement, we mention how we can use it to conclude the proof.
Every $V_{i}$ is a 1-dimensional subspace of $\mathbb{C}^{3}$ and therefore there are only three different types of subspaces. We can now color the nodes in $V(G)$ with three colors, associating each color class to a different subspace.

Let's now prove that either $V_{i} = V_{j}$ or $V_{i} \perp V_{j}$. 
Looking at the adjacency relation between the vertices in the gadget graph $H_{ij}$, we can see that the set $\{\phi(i), \phi(a_{ij}),\phi(b_{ij})\}$ and the set $\{\phi(j), \phi(c_{ij}),\phi(d_{ij})\}$ form a basis for $\mathbb{C}^{3}$. Moreover, we have $\phi(j) = \alpha \phi(i) + \beta \phi(b_{ij})$, $\phi(c_{ij}) = \hat\alpha \phi(i) + \hat\gamma \phi(a_{ij})$ and $\phi(d_{ij}) = \tilde\gamma \phi(a_{ij}) + \tilde\beta \phi(b_{ij})$ where all the constants are complex numbers. 
By contradiction, suppose that $V_{i} \neq V_{j}$  but  $V_{i}$ is not perpendicular to $V_{j}$. This means that both $\alpha$ and $\beta$ are different from zero. Since $\phi(j)$ is orthogonal to $\phi(c_{ij})$, then $\hat\alpha$ must be equal to zero and similarly, since $\phi(j)$ is orthogonal to $\phi(d_{ij})$, also $\tilde\beta$ must be equal to zero.
But now the vector $\phi(c_{ij}) = \hat\gamma \phi(a_{ij})$ is orthogonal to $\phi(d_{ij}) = \tilde\gamma \phi(a_{ij})$, which implies that one of the two vectors is the zero one. This brings a contradiction as $\phi$ is an orthonormal representation of $G'$.
\end{proof}

\begin{proof}[Proof of Theorem~\ref{thm:orthrank}]
The hardness result of Lemma~\ref{lem:NP-xi=3} can be extended to any $k$ greater than $3$ using the suspension operation on graphs as done in~\cite{Laurent:2013}.
Given a graph $G$, the suspension graph $\nabla^t G$ is obtained by adding $t$ new vertices which are pairwise adjacent and are also adjacent to all the vertices in $G$.
Clearly $\xi(\nabla^t G) = \xi(G) + t$ holds.
Hence, for any fixed $k \ge 3$, it is NP-hard to decide whether $\xi(G) \le k$.
\end{proof}

\section{Kremer's Theorem}\label{sec:Kremer}

Here we prove Kremer's Theorem (Theorem~\ref{thm:kremer}), which we restate for convenience. The original proof
by Kremer~\cite{Kremer:1995} applied to boolean functions;
we give a slight generalization of the statement so that it applies to functions with arbitrary range.
It is important to notice that the statements in this section hold for general communication protocols, not only exact ones.

\begin{theorem}
Let~$\ell$ be a positive integer, $X,Y,\mathcal R$ be finite sets and~$\mathcal D\subseteq X\times Y$.
Let $f:\mathcal D\to\mathcal R$ be a function and suppose that~$f$ admits an $\ell$-qubit quantum protocol.
Then, there exists a one-round $2^{O(\ell)}$-bit classical protocol for~$f$.
\end{theorem}

The proof uses the following lemma of Yao~\cite{Yao:1993} and Kremer~\cite{Kremer:1995}.
To reduce the amount of notation needed in the proof we assume that the parties use the following general protocol. At any point during the protocol, both Alice and Bob have a private quantum register. If it is Alice's turn to communicate, say $\ell$ qubits, she appends a fresh $\ell$-qubit register to her existing register, applies a unitary to both registers and sends the $\ell$-qubit register over to Bob, who then absorbs the $\ell$-qubit register into his private register.
If it's his turn to communicate, Bob operates similarly.
This assumption will allow us to deal more easily with protocols in which different numbers of qubits are sent in each round.

\begin{lemma}[Yao--Kremer]\label{lem:yao-kremer}
Let~$\ell$ be a positive integer, $X,Y,\mathcal R$ be finite sets and~$\mathcal D\subseteq X\times Y$.
Suppose that there exists an $r$-round quantum protocol for a function~$f:\mathcal D\to \mathcal R$, where~$\ell_i$ qubits are communicated in round~$i\in[r]$.
Then, the final state of the protocol on input $(x,y)\in\mathcal D$ can be written as
\beqn
\sum \alpha_{\bf u}(x) \beta_{\bf u}(y) \ket{A_{\bf u}(x)}\ket{B_{\bf u}(y)},
\eeqn
where the sum is over all~${\bf u}  \in \bset{\ell_1}\times\cdots\times\bset{\ell_r}$, the~$\alpha_{\bf u}(x), \beta_{\bf u}(y)$ are complex numbers and the $\ket{A_{\bf u}(x)}, \ket{B_{\bf u}(y)}$ are complex unit vectors.
\end{lemma}

\begin{proof}
By induction on~$r$.
The base case~$r = 1$ is trivial, since then Alice sends Bob an~$\ell$-qubit state.
For some~$i\in\{2,3,\dots, r\}$, suppose that after~$i-1$ rounds the state is given by
\beqn
\sum \alpha_{\bf v}(x) \beta_{\bf v}(y) \ket{A_{\bf v}(x)}\ket{B_{\bf v}(y)},
\eeqn
where the sum is over all~${\bf v} \in \bset{\ell_1}\times\cdots\times\bset{\ell_{i-1}}$.
Assume that the $i$-th round is Alice's turn (the case of Bob's turn is handled similarly).
She appends a fresh~$\ell_i$-qubit register to her current register, causing the state to become
\beqn
\sum \alpha_{\bf v}(x) \beta_{\bf v}(y) \ket{A_{\bf v}(x)}\ket{0_1\cdots 0_{\ell_i}}\ket{B_{\bf v}(y)}.
\eeqn
Next, she applies a unitary over both of her registers, turning the state into
\beqn
\sum \alpha_{\bf v}(x) \beta_{\bf v}(y) 
\left(
\sum_{{\bf w}\in\bset{\ell_i}}\gamma_{{\bf w}}\ket{A_{\bf v,w}(x)}\ket{{\bf w}}
\right)
\ket{B_{\bf v}(y)},
\eeqn
where $\gamma_{{\bf w}}$ is a complex number (which might depend on $x$) and for some unit vectors $\ket{A_{\bf v,w}(x)}$. 
Now define
\beqn
\alpha_{\bf v,w}(x) = \alpha_{\bf v}(x)\gamma_{\bf w},
\quad\quad
\beta_{\bf v,w}(y) = \beta_{\bf v}(y)
\quad\quad
\text{and}
\quad\quad
\ket{B_{\bf v,w}(y)}= \ket{{\bf w}}\ket{B_{\bf v}(y)},
\eeqn
so that after the~$i$-th round, after Alice has sent the~$\ell_i$-qubit register to Bob, the state equals
\beqn
\sum_{\bf v,w} \alpha_{\bf v,w}(x) \beta_{\bf v,w}(y) \ket{A_{\bf v,w}(x)}\ket{B_{\bf v,w}(y)}.
\eeqn
After~$r$ rounds the state thus looks like as claimed in the lemma.
\end{proof}

\begin{proof}[Proof of Theorem~\ref{thm:kremer}]
Assume that the protocol proceeds in~$r$ rounds and that $\ell_i$ qubits are communicated during round~$i\in[r]$.
By Lemma~\ref{lem:yao-kremer} the final state of the protocol can be written as
\beqn
\sum \alpha_{\bf u}(x) \beta_{\bf u}(y) \ket{A_{\bf u}(x)}\ket{B_{\bf u}(y)},
\eeqn

To produce his output, Bob performs a measurement~$\{M_1,\dots,M_k\}$ on his register.
For each pair ${\bf u,v} \in \bset{\ell_1}\times\cdots\times\bset{\ell_r}$ and~$j\in[k]$ we define the complex numbers
\beqrn
a_{\bf u,v}(x) &=& \overline{\alpha_{\bf u}(x)}\alpha_{\bf v}(x) \braket{A_{\bf u}(x)}{A_{\bf v}(x)}\\
b_{\bf u,v}^j(x) &=& \overline{\beta_{\bf u}(y)}\beta_{\bf v}(y) \bra{B_{\bf u}(y)}M_j\ket{B_{\bf v}(y)}.
\eeqrn
Then, the probability that Bob gets measurement outcome~$j$ equals
\beqn
p_j(x,y) = \sum_{\bf u,v} a_{\bf u,v}(x) b_{\bf u,v}^j(y).
\eeqn

The classical one-round protocol works in the following way. Let $\ell$ be the total communication of the protocol and define
$\tilde{a}_{\bf u,v}(x)$ as an approximation of $a_{\bf u,v}(x)$ using $2 \ell + 4$ bits for the real part and $2 \ell + 4$ 
bits for the imaginary part, so that $|\tilde{a}_{\bf u,v}(x) - a_{\bf u,v}(x)| \leq 2^{-2 \ell - 3}$.
Alice's message consists of all $2^{2\ell}$ numbers $\tilde{a}_{\bf u,v}(x)$, making the total communication
cost $O(\ell 2^{2\ell})$ bits. Bob calculates his approximation
of the probability of getting outcome $j$ as
\beqn
\tilde{p}_j(x,y) = \sum_{\bf u,v} \tilde{a}_{\bf u,v}(x) b_{\bf u,v}^j(y) .
\eeqn

We can bound the difference between this approximation and the acceptance probability of the original quantum protocol by
\begin{align*}
|\tilde{p}_j(x,y) - p_j(x,y)| & = \Bigl| \sum_{\bf u,v} \bigl( \tilde{a}_{\bf u,v}(x) - a_{\bf u,v}(x) \bigr) b_{\bf u,v}^j(y) \Bigr| \\
& \leq  \sum_{\bf u,v} \bigl| \tilde{a}_{\bf u,v}(x) - a_{\bf u,v}(x) \bigr| \, \bigl| b_{\bf u,v}^j(y) \bigr| \\
& \leq 2^{-2 \ell - 3} 2^{2 \ell}  \leq \frac{1}{8} \,.
\end{align*}

Therefore, given a quantum protocol with sufficiently high success probability,
in this paper in particular probability 1,
Bob can (deterministically) choose the unique outcome $j$ for which $\tilde{p}_j(x,y)$ is strictly greater than $\frac{1}{2}$,
and this outcome $j$ is equal to the function value $f(x,y)$, by correctness of the original quantum protocol.
\end{proof}

\section{Multi-round quantum protocols for \texorpdfstring{\peq-$\binom{n}{\alpha n}$}{EQ-(n an)} with \texorpdfstring{$\alpha < 1/2$}{a < 1/2}}\label{sec:exact-grover-general}

Using distributed versions of Grover's search algorithm, we find multi-round quantum communication protocols that solve the
 \peq-$\binom{n}{\alpha n}$ problem for $\alpha < 1/2$ with a logarithmic number of qubits.
For $\alpha = 1/4$, this statement is proven in Theorem~\ref{thm:exact-grover-n/4}.

When $d = \alpha n$ where $\alpha \in (1/4,1/2)$, we can pad zeros to the inputs such that the new strings are either equal or differ in exactly $1/4$-th of the positions and run the above two-rounds protocol on the new strings. This is the simple idea behind the following theorem.

\begin{theorem}\label{thm:exact-grover>n/4}
For $d = \alpha n$ with $\alpha \in (1/4,1/2)$, the two-round quantum communication complexity of \peq-$\binom{n}{\alpha n}$ is at most $ 2\ceil{\log n} + 2\ceil{\log(4 \alpha)} +1$ qubits.
\end{theorem}

\begin{proof}
Let $x$ and $y$ be Alice's and Bob's inputs.
Both of the players pad the input received with $k = 4 d - n$ zeros.
Hence, the new bit strings $\hat{x}$ and $\hat{y}$ have length $n' = n + k = 4d$ and they are either equal or differ in $n'/4$ positions. 
%
%
%
%
Alice and Bob can now run communication protocol 
  described in the proof of Theorem~\ref{thm:exact-grover-n/4} on the new inputs $\hat{x}, \hat{y} \in \{0,1\}^{n'}$.
The total communication cost is $2\ceil{\log n'} +1 = 2\ceil{\log (4 \alpha n)} +1 \leq 2\ceil{\log n} + 2\ceil{\log(4 \alpha)} +1$ qubits.
\end{proof}

For $d = \alpha n$ where $\alpha \in (0,1/4)$, we need to introduce some technicalities to ensure an exact version of Grover's search algorithm.

\begin{theorem}\label{thm:exact-grover-general}
For $d = \alpha n$ with $\alpha \in (0,1/4)$, the quantum communication complexity of \peq-$\binom{n}{\alpha n}$ is at most 
$O(\log n)$ qubits. 
The
quantum communication protocol uses $O(\frac{1}{\sqrt{\alpha}})$ rounds.
\end{theorem}

\begin{proof}
If a $n$-bit string $z$ is known to contain exactly
$d$ entries that are 1, 
Grover's algorithm can be modified such that it finds an index for one of them with
certainty~\cite[Theorem~16]{Brassard:2002} \cite{Brassard:1998,Ambainis:2004}.
The number of queries $\ell$ that the exact version of Grover's algorithm needs in this case is given by
\[
\ell = \left\lceil{\frac{\pi}{4 \arcsin{\sqrt{\frac{d}{{n}}}}} - \frac{1}{2} } \right\rceil < \frac{\pi}{4} \sqrt{\frac{n}{d}} +1 \,.
\]
The exact version of Grover's algorithm is the same as the original algorithm, except for an adapted final step,
which uses a parametrized diffusion operator $G(\phi)$ and partial
query $V_z(\varphi)$, where $\phi$ and $\varphi$ are angles
that depend on the Hamming distance $d$. 
As these angles do not have a nice closed formula, we refer the reader to~\cite[Equation (12)]{Brassard:2002} for the relation that $\phi$ and $\varphi$ must satisfy.
Here \[ V_z(\varphi) \ket{j} = \begin{cases} \ket{j} & \text{if } z_j = 0 \\ e^{i \varphi} \ket{j} & \text{if } z_j = 1 \end{cases}
\]
and
\[
G(\phi) = F_n V_0(\phi) F_n^{\dagger} \,,
\]
where $F_n$ is the $n \times n$ discrete quantum Fourier transform.\footnote{Note that if $n$ is a power of 2, it is also possible to use the $n\times n$ Hadamard transform.} 

Take $x,y \in \{0,1\}^n$ to be the input strings of  Alice and Bob and let
$z = x \oplus y$.
As in the proof of the $n/4$ case of Theorem~\ref{thm:exact-grover-n/4}, we turn this search algorithm into a quantum communication protocol by writing
a single query $U_z = U_x U_y = U_y U_x$. 
We can use the commutativity of $U_x$ and $U_y$ to save rounds:
The exact Grover's algorithm is performed by executing the operations
\[
 G(\phi) V_z(\varphi) \underbrace{G U_z \ldots G U_z}_{\ell-1 \text{ times}}
\]
on starting state $\ket{s} = \frac{1}{\sqrt{n}} \sum_{i=1}^{n} \ket{i}$. Since we can write two
alternations as $G U_z G U_z = G U_x U_y G U_y U_x$, alternating whether
Alice or Bob executes the query first that round,
only $\ell-1$ rounds are needed for the $\ell-1$
ordinary Grover iterations. Alice starts the protocol if $\ell$ is even,
and Bob sends the first message if $\ell$ is odd.

For the final step, the players need to simulate a query $V_z(\varphi)$ by
local operations that depend only on $x$ or $y$. At this point
in the protocol it is Alice's turn to communicate. She currently
holds the state 
\[
\ket{\psi} = \underbrace{G U_z \ldots G U_z}_{\ell-1 \text{ times}} \ket{s} \,.
\]

Now Alice adds an auxiliary qubit that starts in state $\ket{0}$.
Define the unitary operation $Q_x$ by its action on the computational basis states
as
\[Q_x \ket{j} \ket{b} = \ket{j} \ket{b \oplus x_j}
\]
and the (diagonal) unitary matrix $R_y(\varphi)$ as 
\[
R_y(\varphi) \ket{j} \ket{b} = e^{i \varphi (b \oplus y_j)} \ket{j} \ket{b} \,.
\]
Now Alice first applies $Q_x$ on the state $\ket{\psi}\ket{0}$, sends
this state to Bob who performs $R_y(\varphi)$, sending the state back
to Alice who again performs $Q_x$. It is easy to check that $Q_x R_y(\varphi) Q_x \ket{\psi} \ket{0} = (V_z(\varphi) \otimes I) \ket{\psi} \ket{0}$, therefore Alice now discards the auxiliary qubit
and applies $G(\phi)$ to finish the simulation of the exact version of Grover's algorithm.

The final state of the exact Grover's algorithm is
$\frac{1}{\sqrt{d}} \sum_{i \text{ s.t.~} z_i=1} \ket{i}$ if $|z|=d$.
If Alice has this state in her possession, she performs a measurement in the computational basis, obtaining
an index $i^*$ such that $x_{i^*} \neq y_{i^*}$ if $x \neq y$.
Then she sends $i^*$ and the value $x_{i^*}$ over to Bob, who outputs `equal' if and only if $x_{i^*} = y_{i^*}$. This final message
consists of $\ceil{\log{n}}+1$ qubits.
By the correctness of the exact Grover's algorithm, this protocol correctly outputs `not equal' if the Hamming distance between $x$ and $y$ is the
fixed value $d$. 
Therefore we turned a $\ell$-query execution of the exact version of Grover's algorithm into a protocol that uses $(\ell+2) \ceil{\log {n}} + 2$ qubits of communication in $\ell+2$ rounds. 
\end{proof}

\section{Bound on binary entropy function}\label{app:entropy}

Here we prove the following simple lemma, whose statement we have used in the proof of Theorem~\ref{thm:xi-smaller12}.
Recall that the binary entropy function $H$ is defined as $H(p) = - p \log p - (1-p)\log(1-p)$ for $p \in [0,1]$.

\begin{lemma}\label{lem:entropy}
For any $p \in (0, 1/2)$, we have $H(p) + H ( 1/2 - \sqrt{(1-p)p} ) - 1 > 0$.
\end{lemma}

The proof of the lemma uses the following lower bound for the binary entropy function.

\begin{lemma}\label{lem:H-bound}
For any $p \in [0,1]$, $H(p) \geq 1 - (1-2p)^{2}$ holds. Moreover, equality holds if and only if $p \in \{0,1/2,1\}$.
\end{lemma}

\begin{proof}
The Taylor series of the binary entropy function around the point $1/2$ gives that 
$$1 - H(p) = \frac{1}{2 \ln 2} \sum_{n = 1}^{\infty} \frac{(1-2p)^{2n}}{n(2n-1)} 
\leq \frac{(1-2p)^{2}}{2 \ln 2} \sum_{n = 1}^{\infty} \frac{1}{n(2n-1)}
 = (1-2p)^{2},
 $$
where the first inequality is due to the fact that $|1-2p| \leq 1$ and therefore $(1-2p)^{2n} \leq (1-2p)^{2}$
and the last equation is due to fact that $2 \ln 2 = \sum_{n \ge 1} \frac{1}{n(2n-1)}$.
Indeed, the Taylor series for $\ln 2$ around $0$ (also known as Mercator series) gives that $\ln 2 = \sum_{n \ge 1} \frac{(-1)^{n+1}}{n} = \sum_{n \ge 1} \frac{1}{2n(2n-1)}$, and multiplying both sides by 2 gives the wanted result.

Therefore, we deduce that $H(p) \geq 1 - (1-2p)^{2}$. Moreover, equality holds only at the points where $(1-2p)^{2n} =(1-2p)^{2}$, which are $p \in \{ 0,1/2,1 \}$.
\end{proof}

\begin{proof}[Proof of Lemma~\ref{lem:entropy}]
Using Lemma~\ref{lem:H-bound}, we have that for any $p \in (0,1/2)$ the following holds:
$$
H(p) + H ( 1/2 - \sqrt{(1-p)p} ) - 1 > 1 - (1-2p)^{2} + 1 - (1-1 +2\sqrt{(1-p)p})^{2} - 1 = 0.
$$
\end{proof}

\section{Partial results on \texorpdfstring{$\xi(G_{\mathcal K})$}{xi(G\_K)}}\label{sec:app-orth}

We collect here some small results about the orthogonal rank of the graph $G_{\mathcal K}$.
In particular, we give an upper and lower bound on $\xi(G_{\mathcal K})$ (Proposition~\ref{prop:uppxi} and Proposition~\ref{prop:lowxi}, respectively) and show that a variant of the Lov\'asz theta number 
 does not provide an useful lower bound to $\xi(G_{\mathcal K})$ (Proposition~\ref{prop:thetaG_{k}}).

Recall that for $\mathcal K = \mathcal L_{n/2}\cup \cdots \cup \mathcal L_{n}$, we defined $G_{\mathcal K} = (\{0,1\}^{n},E)$ to be the graph whose edge set $E$ consists of all pairs of strings with Hamming distance $\{n/2,\dots,n\}$.
Also recall the definition the binary entropy function, $H(p) = - p \log p - (1-p)\log(1-p)$.

\begin{proposition}\label{prop:uppxi}
The graph $G_{\mathcal K}$ as above satisfies $\xi(G_{\mathcal K}) \leq 2^{H(1/4) n + 1} \approx 2^{0.81 n + 1}$. 
\end{proposition}

\begin{proof}
To start, we show that the existence of a certain polynomial implies the existence of an  orthonormal representation of the graph.
Let $P : \{-1,1\}^n \to \mathbb{R}$ be a polynomial with non-negative real coefficients on $n$ variables $z_1, \ldots, z_n$, each in $\{-1,1\}$.
That is, 
\[
 P(z) = \sum_{S \subseteq [n]} \alpha_S \prod_{k \in S} z_k, \text{ with } \alpha_S \geq 0,
\]
where for brevity we use $z$ as a shorthand for the vector $z_1, z_2, \ldots, z_n$.
Let $\mon P = \bigl| \{ S : \alpha_S \neq 0 \} \bigr|$ denote the number of monomials of $P$ with a non-zero coefficient.

Note that if we have a polynomial $P$ such that $P(1,1,\ldots, 1) \neq 0$ and $P(z)=0$ for all $z \in \{-1,1\}^n$ such that $\sum_{k=1}^n z_k = 2d - n$ then
we can turn this into a $\mon P$-dimensional orthonormal representation of the graph $H(n, d)$ as following.
Label the canonical basis vectors by the sets $S$ for which $\alpha_S \neq 0$, and take the map $\phi : \{0,1\}^n \to \mathbb{C}^{\mon P} $ defined by 
\[
\phi(x) = \sum_{S : \alpha_S \neq 0} \sqrt{\alpha_S} \prod_{k \in S} (-1)^{x_k} e_S \,. 
\]
For two strings $x,y$ the inner product of the associated vector is then given by
$\langle \phi(x), \phi(y) \rangle = \sum_{S : \alpha_S \neq 0} \alpha_S \prod_{k \in S} (-1)^{x_k + y_k} = P(z)$
for $z_k = (-1)^{x_k + y_k}$. The orthogonality then follows by noting that if $x$ and $y$ differ in $d$ positions,
 we have $\sum_{k=1}^n z_k =  \sum_{k=1}^n (-1)^{x_k + y_k} = \bigl| \{k : x_k = y_k\} \bigr| - \bigl| \{k : x_k \neq y_k\} \bigr| = n-2d$.

We will find an orthonormal representation for $G_{\mathcal K}$ by constructing a polynomial $P$ such that $P(z) = 0$ for $-n \leq \sum_{k=1}^n z_k \leq 0$.
Now, for any even $d$ with $n/2 \leq d \leq n$ define the polynomial 
\[
  P_d(z) = (2d-n) + \sum_{k=1}^n z_k \, .
\]
The evaluation of $P_d$ is 0 whenever $\sum_{k=1}^n z_k = n - 2d$ for the given $d$. Also define the polynomial
\[
P_{odd}(z) = 1 + \prod_{k=1}^n z_k    \, .
\]

This polynomial equals $0$ whenever an odd number of $z_k$ equal $-1$. To simplify notation, define $c_k = 4k$. Multiplying the $P_d$ together gives a polynomial which is 0 on the correct inputs for all even $d \geq n/2$:
\begin{align*}
 P_{even}(z) &= \prod_{\substack{d = n/2 \\ d \text{ even}}}^n P_d(z) = \prod_{\substack{d = n/2 \\ d \text{ even}}}^n (2d-n + \sum_{j=1}^n z_j) =  \prod_{k=0}^{n/4} (c_k + \sum_{j=1}^n z_j) \\
 &= \sum_{k=0}^{n/4} \biggl[ \sum_{\substack{S \subset [n/4] \\ |S| = k}}  \Bigl(\prod_{j\in S} c_j \Bigr)  \Bigl(\sum_{j=1}^{n} z_j \Bigr)^{n/4-k} \biggr] \,.
\end{align*}
Since this polynomial has only monomials with degree at most $n/4$,
we can upper bound the number of monomials by       
${\mon P_{even} \leq \sum_{k=0}^{n/4} \binom{n}{k} \leq 2^{H(1/4) n}}$.
The product $P(z) = P_{even}(z) P_{odd}(z)$ gives an orthonormal representation with the desired properties, and since $\mon P_{odd} = 2$, we have that $\mon P \leq 2^{H(1/4) n + 1}$.
\end{proof}

To the best of our knowledge, the current-best lower bound on~$\xi(G_{\mathcal K})$ is linear in~$n$.
The argument presented here is due to Alon~\cite{Alon:personal}.
We use the following result due to Kleitman~\cite{Kleitman:1966}.
The \emph{diameter} of a set~$A\subseteq \bset{n}$ is defined as the maximum Hamming distance between any two elements in~$A$.

\begin{theorem}[Kleitman]\label{thm:kleitman}
Let $r$ and $n$ be positive integers such that~$r\leq n/2$.
Let~$A\subseteq \bset{n}$ be a set of diameter~$2r$.
Then,
\beqn
|A| \leq \sum_{k=0}^r \binom{n}{k} \leq 2^{nH(r/n)}.
\eeqn
\end{theorem}

\begin{proposition}\label{prop:lowxi}
The graph $G_{\mathcal K}$ as above satisfies $\xi(G_{\mathcal K}) \geq \Omega(n)$.
\end{proposition}

\begin{proof}
Suppose that there exists an orthogonal representation of $G_{\mathcal K}$ into~$S^{d-1}$.
Observe that any spherical cap of $S^{d-1}$ of angle strictly less than $\pi/2$ contains only points of the representation that correspond to an independent set in~$G_\mathcal K$.
For $\eps > 0$, consider a covering of $S^{d-1}$ by spherical caps of radius $\pi/2 - \eps$ using the minimum number of caps.
A simple volume argument 
shows that for some absolute constant~$\alpha = \alpha(\eps) \in (0,\infty)$, such a covering uses at most $2^{\alpha d}$ caps.
By the Pigeonhole Principle, there exists a cap in the covering and a set $A\subseteq\bset{n}$ of cardinality at least $2^{n - \alpha d}$, such that the representation sends each string in~$A$ to the cap.
Since~$A$ must be an independent set in $G_\mathcal K$, it has diameter less than~$n/2$.
Hence, by Theorem~\ref{thm:kleitman}, we have $n - \alpha d\leq nH(1/4)$, giving the result.
\end{proof}

At last, we show that a variant of the Lov\'asz theta number cannot be used to get at stronger lower bound on $\xi(G_{\mathcal K})$ than the one presented in the above Proposition~\ref{prop:lowxi}.
The variant we consider, called~$\vartheta'(G)$, was introduced in~\cite{Schrijver:1979,McEliece:1978}  by adding nonnegativity constraints to the maximization program that defined $\vartheta(G)$ in~(\ref{eq:thetadual}):
\beqn
\vartheta'(G) = \max \{ \sum_{i,j \in [n]} X_{ij} \st X \in S_{+}, \; \sum_{i \in [n]} X_{ii} = 1, \; X_{ij} \geq 0 \; \forall \, ij \in [n] \text{ with equality if $ij \in E(G)$} \}.
\eeqn
 Clearly, for any graph $G$, the following chain of inequality holds: $\vartheta'(\overline{G}) \leq \vartheta(\overline{G}) \leq \xi(G)$. 
 In the following proposition, we show that $\vartheta'(\overline{G_{\mathcal{K}}}) \leq 2n$, meaning that this parameter does not provide a useful lower bound on $\xi(\overline{G_{\mathcal{K}}})$.
Although we don't know whether $\vartheta(\overline{G_{\mathcal{K}}})$ is also upper bounded by $cn$, for some $c$ constant, computational results seems to indicate that these two parameters are always close to each other~\cite{Meurdesoif:2005,DR:2008}. 
 
\begin{proposition}\label{prop:thetaG_{k}}
The graph $G_{\mathcal K}$ as above satisfies $\vartheta'( \overline{G_{\mathcal K}}) \leq 2n$.
\end{proposition}

\begin{proof}
We consider the complement graph $\overline{G_{\mathcal K}}$, whose vertices are all the strings in $\{0,1\}^{n}$ and two strings are adjacent if their Hamming distance is at most $n/2 - 1$. 
This graph arises in the context of Hamming schemes (see~\cite{Delsarte:1998} for a background on association schemes) and Schrijver~\cite{Schrijver:1979} (see also~\cite{McEliece:1978}) proved that both $\vartheta(\overline{G_{\mathcal K}})$ and $\vartheta'(\overline{G_{\mathcal K}})$ can be written as a linear program. In particular, 
$$
 \vartheta'(\overline{G_{\mathcal K}})  =  
 \max \bigg\{ 1+ \sum_{k =n/2}^{n} a_{k}  \st  a_{k} \geq 0 \; \forall \, k \in \{n/2,\dots,n\};  
 \binom{n}{d}+ \sum_{k=n/2}^{n} a_{k} K_{d}^{n}(k) \geq 0 \; \forall \, d \in \{0,\dots,n\} \bigg\}
$$
where $K_{d}^{n}(k)$ is the evaluation at $k$ of the degree-$d$ (binary) Krawtchouk polynomial $K_{d}^{n}$ defined as~(\ref{Krawtchouk}).
This program is equal to the well-known Delsarte's linear program~\cite{Delsarte:PhD}.

Let us now consider a new linear program where, among the constraints using the Krawtchouk polynomials, we only keep the ones of degree-1:
$$ \lambda(\overline{G_{\mathcal K}})  :=  
\max  \bigg\{ 1+ \sum_{k =n/2}^{n} a_{k}  \st  a_{k} \geq 0 \; \forall \, k \in \{n/2,\dots,n\};  \;
 n+ \sum_{k=n/2}^{n} a_{k} K_{1}^{n}(k) \geq 0  \bigg\}. $$
Clearly, $\vartheta'(\overline{G_{\mathcal K}}) \leq \lambda(\overline{G_{\mathcal K}})$ and it can be easily derived that the constraints of $\lambda(\overline{G_{\mathcal K}})$ imply the Plotkin bound, i.e., that $\lambda(\overline{G_{\mathcal K}}) \leq 2n$ (see for example~\cite[Section 4.3]{Delsarte:PhD}). 
\end{proof}

\end{document}